\documentclass[a4paper]{article}
\usepackage{amsmath}
\usepackage{amssymb}
\usepackage{amsthm}
\usepackage{graphicx}
\usepackage{color}

\theoremstyle{definition}

\theoremstyle{plain}
\newtheorem{theorem}{Theorem}

\newtheorem{corollary}{Corollary}
\newtheorem{lemma}{Lemma}
\theoremstyle{remark}
\newtheorem{case}{Case}

\newcommand{\eq}[1]{Eq.\ (\ref{eq:#1})}
\newcommand{\Lc}{\mathcal{L}} 
\newcommand{\ud}{\,\mathrm{d}}
\newcommand{\N}{{\mathbb{N}}}
\newcommand{\R}{{\mathbb{R}}}

\newcommand{\E}[2]{{\mathbb{E}_{#1} [#2]}}

\newcommand{\feta}{{\boldsymbol{\eta}}}

\newcommand{\rhobg}{\rho_{\textrm{bg}}}

\newcommand{\rhotrans}{\rho_{\textrm{trans}}}

\newcommand{\SLbg}{S_L^{\textrm{bg}}}

\newcommand{\fig}[1]{Fig.\ \ref{fig:#1}}
\newcommand{\sect}[1]{Section \ref{sec:#1}}

\graphicspath{{./images_v2/}}
\title{Finite size effects and metastability in zero-range condensation}
\author{Paul Chleboun$^*$,\ Stefan Grosskinsky\footnote{Mathematics Institute and Centre for Complexity Science, University of Warwick, Coventry, CV4 7AL, UK }}
\date{\today}

\begin{document}
\maketitle

\begin{abstract}
We study zero-range processes which are known to exhibit a condensation transition, where above a critical density a non-zero fraction of all particles accumulates on a single lattice site. This phenomenon has been a subject of recent research interest and is well understood in the thermodynamic limit. The system shows large finite size effects, and we observe a switching between metastable fluid and condensed phases close to the critical point, in contrast to the continuous limiting behaviour of relevant observables. We describe the leading order finite size effects and establish a discontinuity near criticality in a rigorous scaling limit. We also characterise the metastable phases using a current matching argument and an extension of the fluid phase to supercritical densities. This constitutes an interesting example where the thermodynamic limit fails to capture essential parts of the dynamics, which are particularly relevant in applications with moderate system sizes such as traffic flow or granular clustering.


\end{abstract}

keywords: zero range process; condensation; metastability; finite size effects; large deviations\\

\section{Introduction}
\label{sec:intro}
Zero-range processes are stochastic particle systems with no restriction on the number of particles per site and with jump rates that depend only on the occupation of the departure site. 
This simple zero-range interaction leads to a product structure of the stationary distributions \cite{SpitzerInteraction,AndjelInvariant}. 
These processes have been a focus of recent research interest since they can exhibit a condensation transition. 
This is the case for space homogeneous jump rates $g(n)$ that decay asymptotically with the number of particles $n$.
A prototypical model with $g(n)=1+\frac{b}{n^\gamma}\quad\mbox{for }n=1,2,\ldots$ has been introduced in \cite{EvansPhase}, where condensation occurs for parameter values $\gamma \in (0,1)$, $b>0$ or $\gamma =1$, $b>2$. 
If the particle density $\rho$ exceeds a critical value $\rho_c$, the system phase separates into a homogeneous background with density $\rho_c$ and a condensate, where the excess particles accumulate on a single randomly located lattice site. 
This transition has been established on a rigorous level in a series of papers, in the thermodynamic limit  \cite{JeonSize,GrosskinskyCondensation,ArmendarizThermodynamic,ArmendarizZero}, as well as on a finite system as the total number of particles diverges \cite{FerrariCondensation}.
Dynamic aspects of the transition such as equilibration and coarsening \cite{GodrecheDynamics,GrosskinskyCondensation} and the stationary dynamics of the condensate \cite{GodrecheDynamicsa} are well understood heuristically, for the latter first rigorous results have been achieved recently \cite{BeltranMetastability,BeltranMeta-stability}.
Findings for the zero-range process could be applied to understand condensation phenomena in a variety of nonequilibrium systems (see \cite{EvansNonequilibrium} and references therein), as well as providing a generic model of domain wall dynamics and a criterion for phase separation using a mapping to one-dimensional exclusion systems \cite{KafriCriterion}.
The process continues to be of interest, recent work on variations of the model includes mechanisms leading to more than one condensate \cite{SchwarzkopfZero-range,ThompsonZero-range,KimParticle}, or the effects of memory in the dynamics \cite{HirschbergCondensation}.

While most of the results so far consider the thermodynamic limit, finite size effects in the model with jump rates $g(n)$ and $\gamma =1$ have been investigated in \cite{EvansCanonical,MajumdarNature} using saddle point methods, and in \cite{AngelCondensation} for a variant of the model with a single defect site. 
In the condensed phase region, finite systems are found to exhibit a large overshoot of the stationary current above its value in the thermodynamic limit.
In this paper we examine this phenomenon in detail for all possible values of the parameter $\gamma\in (0,1]$. 
We find the leading order finite size effects that describe the current overshoot by continuing the homogeneous (fluid) phase above the critical density and characterizing the condensed phase by a current matching argument. 
For $\gamma <1$, the main focus of this article, finite systems exhibit a metastable switching behaviour between the two phases, which is prevalent in Monte Carlo simulations for a wide range of parameters. 
To capture this phenomenon we examine the system in a scaling limit and derive a rate function, which exhibits a double well structure, describing the distribution over the bulk density (the bulk density serves as an order parameter to distinguish the fluid and condensed phases).
This way we rigorously establish the discontinuous behaviour on the critical scale, even though the bulk density is a continuous function in the thermodynamic limit. 
This is shown to be in agreement with recent results on the condensation transition at the critical density \cite{ArmendarizZero}. 
Based on the exact scaling limit, we can also predict the lifetime of the metastable phases for large finite systems by a heuristic random walk argument. 

In general, finite system size can lead to effective long-range interactions and non-convexity of thermodynamic potentials, as has been observed for various models (see for example \cite{HuellerFinite,BehringerContinuous} and references therein). For the zero-range process it has  been shown that for simple size-dependent jump rates metastability effects can be manifested even in the thermodynamic limit \cite{GrosskinskyDiscontinuous}, and large crossover effects in these systems have already been observed in \cite{EvansPhase}. 
The finite-size behaviour of the zero-range process considered here also exhibits the above non-convexity, with the additional feature of a sharp crossover between a putative fluid and condensed phase in the non-convex part. 
This leads to a metastable switching behaviour which disappears in the thermodynamic limit, intriguingly contradicting the usual expectation of finite systems to behave in a smoother fashion than the limiting prediction. 
The onset of phase coexistence at criticality is a classical question of general interest in phase separating systems, see for example \cite{BiskupFormation,BiskupCritical} for the formation of equilibrium droplets in the Ising model which also form suddenly on a critical scale. 
Rigorous results on metastability regarding the dynamics of the condensate in zero-range processes have also been a subject of recent research interest \cite{BeltranMetastability,BeltranMeta-stability,BeltranTunneling}, and our work provides a contribution in that direction and new insight in the mechanisms of condensate dynamics on finite systems.
This is explained in more detail in the discussion.

A proper understanding of the metastability phenomenon exhibited by zero-range processes on finite systems is also of particular importance for recent applications with moderate system sizes. 
Clustering phenomena in granular media can be described by zero-range processes (see \cite{ToeroekAnalytic,MeerCompartmentalized} and references therein), and metastable switching between homogeneous and condensed states has been observed experimentally \cite{WeeleHysteretic,MeerSudden}. 
In the spirit of the mapping introduced in \cite{KafriCriterion} the zero-range process with jump rates $g(n)$ has also been applied as a simplified traffic model \cite{KaupuzsZero-range,LevineTraffic}, where condensation corresponds to the occurance of a traffic jam. 
A key feature of traffic models is the existence of a broad range of densities over which metastability between free flowing and jammed states is observed (see for example \cite{WilsonMechanisms} and references therein).
The relevance of this study in application serves as a motivation, but the aim of the paper is a general understanding of finite size effects and their implications for a generic class of zero-range processes, rather than a detailed analysis of particular cases.

The paper is organized as follows: 
In Section 2 we introduce the model and summarise previous results on stationary distributions and the thermodynamic limit. 
In Section 3 we present relevant observations on finite systems which we later study in detail. 
In Section 4 we give preliminary results and a heuristic description of the finite size effects. 
This motivates our main results presented in Section 5, where we rigorously establish metastability in a suitable scaling limit. 
In Section 6 we connect these results to the lifetimes of the metastable phases, and end with a short discussion in Section 7.

\section{Stationary Measures and the Thermodynamic Limit}
\subsection{The Zero Range Process}
We consider a one dimensional lattice of L sites $\Lambda_L = \{1,2,\ldots,L \}$ with periodic boundary conditions.
Let $\eta_x\in\N_0=\{0,1,2,\ldots\}$ be the number of particles on site $x\in\Lambda_L$.
The state of the system is described by $\feta=(\eta_x)_{x\in\Lambda_L}$ belonging to the state space of all particle configurations  $\Omega_L = \N_0^{\Lambda_L}$.
Particles jump on the lattice at a rate that depends only on the occupation number of the departure site.
A particle jumps off site $x$ after an exponential waiting time with rate $g(\eta_x)$ and moves to a target site $y$ according to the probability distribution $p(y-x)$.
We assume that $p$ is of finite range, i.e.\ $p(z) = 0$ if $|z|>R$ for some $R>0$, normalised and irreducible on $\Lambda_L$.
The main results of the paper focus on the jump rates $g:\N_0\to\R_{+}$ of the form
\begin{align}
  \label{eq:rates}
g(n) = \left\{
    \begin{array}{ll}
      1 + \frac{b}{n^\gamma} &\textrm{if $n>0$} \\
      0 &\textrm{if $n = 0$}
    \end{array}
    \right.    
\end{align}
with $\gamma\in(0,1)$ and $b>0$. A discussion on extending these results to the case $\gamma=1$ and $b\geq 3$ and other lattice geometries can be found in section \ref{sec:discuss}. 
These jump rates were first introduced by Evans \cite{EvansPhase} and represent a fairly general class of functions of interest for the condensation transition.

The infinitesimal generator of the process acting on suitable test functions $f$ is given by
\begin{align}
  \label{eq:infgen}
  (\Lc f)(\feta) = \sum_{x,y\in\Lambda_L}g(\eta_x)p(y-x)\left(f(\feta^{x,y})-f(\feta)\right)\ ,
\end{align}
where $\eta^{x,y}_z = \eta_z - \delta(z,x) + \delta(z,y)$ and $\delta$ is the Kronecker delta \cite{AndjelInvariant,LiggettErgodic}. 
The process conserves the total number of particles in the system, so 
$\Omega_L$ can be partitioned into invariant subsets $\Omega_{L,N} = \left\{\feta\in\Omega_L | \sum_{x\in\Lambda_L}\eta_x = N\right\}$
on which the zero-range process is a finite state irreducible Markov process.
The process can also be defined on an infinite lattice under certain constraints, for details see \cite{HolleyClass,AndjelInvariant}.

\subsection{Stationary Measures}
\label{subsec:statmeasure}
The following summarises well known results on stationary measures of the zero-range process, for details see \cite{SpitzerInteraction,EvansPhase,AndjelInvariant}. 
The zero-range process with generator (\ref{eq:infgen}) has a family of stationary homogeneous product measures on $\Omega_{L}$ which we refer to as the \textbf{grand canonical ensemble}. 
These measures are parameterized by a fugacity $\phi$ and are of the form,
\begin{align}
  \nu_{\phi}^{L}\left[\feta\right]=\prod_{x\in\Lambda_L}\nu_{\phi}\left[\eta_x\right] \quad \textrm{where} \quad \nu_\phi\left[n\right] = \frac{1}{z(\phi)}w(n)\phi^{n} .
\end{align}
These exist for all $\phi \in [0,\phi_c)$ where $\phi_c$ is the radius of convergence of the (single site) partition function
\begin{align}
  z(\phi) = \sum_{k=0}^{\infty}w(k)\phi^k .
\end{align}
The stationary weights $w$ are given by $w(0)=1$ and
\begin{align}
\label{eq:wn}
  w(n) = \prod_{k=1}^{n}g(k)^{-1}, \quad n>0 .
\end{align}

In the grand canonical ensemble the expected particle density is a function of $\phi$ and is given by,
\begin{align}
  R(\phi) := \E{\nu_{\phi}}{\eta_1} = \sum_{k=0}^{\infty}k\nu_\phi(k) = \phi\,\partial_\phi\log z(\phi)\ ,
\end{align}
which is strictly increasing and $R(0)=0$. 
The critical density is defined by $\rho_c = \lim\limits_{\phi\nearrow\phi_c}R(\phi)\in(0,\infty]$, and condensation occurs if $\rho_c < \infty$, as explained below.

The expected jump rate off a site is proportional to the average stationary current or, in case the first moment $\sum_z z\, p(z)$ vanishes, to the diffusivity. 
Therefore for simplicity, in the rest of this paper, \textbf{current} will refer to the average jump rate off a site, which is clearly site independent under a stationary distribution.
In the grand canonical ensemble the current is simply given by the fugacity $\phi$, 
\begin{align}
\label{eq:GCcur}
  j_\phi := \E{\nu_{\phi}}{ g(\eta_x)} &= \frac{1}{z(\phi)}\sum_{n=0}^{\infty}g(n) w(n) \phi^{n}=\phi\ ,
\end{align}
which follows directly from the form of the stationary weights $w(n)$ \eq{wn}.

For fixed $L$ and $N$ the process restricted to $\Omega_{L,N}$ is ergodic, the corresponding unique stationary measures belong to the \textbf{canonical ensemble} and are given by
\begin{align}
  \pi_{L,N} \left[\feta \right]:=\nu_\phi\left[\feta | S_L(\feta) = N\right] \quad \textrm{where} \quad S_L(\feta) = \sum_{x\in \Lambda_L} \eta_x.
\end{align}
These measures are independent of $\phi$ and are given explicitly by
\begin{align*}
  \pi_{L,N} \left[\feta \right]=\frac{1}{Z(L,N)}\prod_{x\in\Lambda_L}w(\eta_x)\delta\big(\sum_x \eta_x , N\big)\ ,
\end{align*}
where the canonical partition function is
\begin{align}
  Z(L,N) = \sum_{\feta\in\Omega_{L,N}}\prod_{x\in\Lambda_L}w(\eta_x)\ .
\end{align}

In the canonical ensemble the form of the stationary weights, \eq{wn}, imply that the average current is given by a ratio of partition functions,
\begin{align}
  \label{eq:CanCur}
  j_{L,N}:=\E{\pi_{L,N}}{ g(\eta_x) } = \frac{Z(L,N-1)}{Z(L,N)}.
\end{align}

For the jump rates (\ref{eq:rates}) the single site weights asymptotically decay as a stretched exponential
\begin{align}
  \label{eq:AsymWk}
  w(n)\sim e^{-\frac{b}{1-\gamma}n^{1-\gamma}}  \quad \textrm{for} \quad \gamma \in (0,1).
\end{align}
Throughout the paper we use `$\sim$' to mean asymptotically proportional and `$\simeq$' to mean asymptotically equal.
It follows that $\phi_c=1$ and the critical density and the variance are finite, 
\begin{align*}
  \rho_c := R(1) < \infty\ ,\quad  \sigma_c^2 :=\E{\nu_{\phi_c}}{\eta_x^2} - \rho_c^2 < \infty\ .
\end{align*}
$R$ is strictly increasing, so invertible on $[0,\phi_c]=[0,1]$ and we denote its inverse by
\begin{align}
  \tilde{\Phi}(\rho)=R^{-1}(\rho).
\end{align}
In this way we can parameterise the grand canonical measures by densities $\rho\in [0,\rho_c]$.

\subsection{Thermodynamic Limit}
\label{sec:thermo}
In this section we summarise known results on the condensation transition in the zero-range process.
It has been established  in \cite{EvansPhase} and rigorously in \cite{GrosskinskyCondensation,GrosskinskyEquivalence} as a continuous phase transition in the thermodynamic limit,
as particle number $N$ and lattice size $L$ tend to infinity such that $N/L\to\rho$.
It was shown that all finite dimensional marginals of the canonical measure $\pi_{N,L}$ converge to the grand canonical measure with density $\rho$ if $\rho\leq \rho_c$.
If $\rho>\rho_c$ then there is no grand canonical measure with density $\rho$ and all finite dimensional marginals of $\pi_{L,N}$ converge to the grand canonical measure with density $\rho_c$.
This result on the equivalence of ensembles holds in terms of weak convergence (for details see \cite{GrosskinskyEquivalence}) and can be summarised as
\begin{align}
  \label{eq:weak}
  \pi_{L,N} \to \nu_{\Phi(\rho)} \quad  \textrm{as} \ L\to\infty \quad \textrm{and} \quad N/L\to\rho\ ,
\end{align}
where
\begin{align}
  \label{eq:phi}
  \Phi(\rho) = \left\{ \begin{array}{ll}
       \tilde{ \Phi}(\rho) & \textrm{ if $\rho<\rho_c$}\\
        \phi_c & \textrm{ if $\rho \geq \rho_c$ }
      \end{array}     \right. .
\end{align}
This implies that for $\rho>\rho_c$ the excess particles in the system condense on a set of vanishing volume fraction. 
The result has been strengthened in \cite{ArmendarizThermodynamic,ArmendarizZero} (with partial results already in \cite{JeonSize}), showing that the condensate typically resides on a single lattice site. 
Denoting the size of the maximum component $M_L (\feta) = \max\limits_{x\in \Lambda_L} \eta_x$ this can be written as
\begin{align}
\frac{1}{L} M_L \xrightarrow{\pi_{L,N}} \left\{ \begin{array}{ll}
       0 & \textrm{ if $\rho\leq\rho_c$}\\
        \rho - \rho_c & \textrm{ if $\rho > \rho_c$ }
      \end{array}     \right. .
\end{align}
Here the notation denotes converges in probability with respect to $\pi_{L,N}$,
\begin{align*}
  \pi_{L,N}\left[\left|\frac{1}{L} M_L - (\rho-\rho_c)\right|>\epsilon\right] \to 0 \quad \textrm{for all} \quad \epsilon>0.
\end{align*}
It will often be useful to express this result also in terms of the behaviour of the bulk of the system. 
We denote the total number of particles outside of the maximally occupied site by $\SLbg(\feta) = N - M_L(\feta)$.
The background density converges as
\begin{align}
\label{eq:Thermobg}
\frac{1}{L-1}\SLbg(\feta) \xrightarrow{\pi_{L,N}} \left\{ \begin{array}{ll}
       \rho & \textrm{ if $\rho\leq \rho_c$}\\
        \rho_c & \textrm{ if $\rho >\rho_c$ }
      \end{array}     \right. .
\end{align}
A similar result holds for the behaviour of the canonical current in the thermodynamic limit.
With \eq{GCcur} the convergence in (\ref{eq:weak}), (\ref{eq:phi}) implies that
\begin{align}
  j_{L,N}\to \Phi(\rho) \quad \textrm{as }L\to\infty \quad \textrm{and} \quad N/L\to \rho\ .
\end{align}
So the current and the background density are both continuous with respect to the total system density in the thermodynamic limit. 
Also both are strictly increasing up to $\rho_c$ and constant for $\rho>\rho_c$. 
If $N/L\to\rho\leq\rho_c$ the system is said to be in the \textbf{fluid} phase region and if $N/L\to\rho > \rho_c$ the system is in the \textbf{condensed} phase region.

The thermodynamic entropy is defined by the Legendre transform
\begin{align}
  \label{eq:thermoEnt}
    s(\rho) &= \sup\limits_{\phi\in[0,\phi_c)}\left( \log z (\phi) - \rho \log \phi \right)\nonumber \\
    &= \log z (\Phi(\rho)) - \rho \log \Phi(\rho) \ ,
\end{align}
where $\Phi(\rho)$ is given by \eq{phi}.
It has been shown in \cite{GrosskinskyCondensation} that the canonical partition function converges to the thermodynamic entropy, i.e.
\begin{align}
    \frac{1}{L}\log Z(L,N) \to s(\rho) \quad \textrm{as }L\to\infty \quad \textrm{and} \quad N/L\to \rho\ .
\end{align}

\section{Observations on a Finite Systems}

\begin{figure}[t]
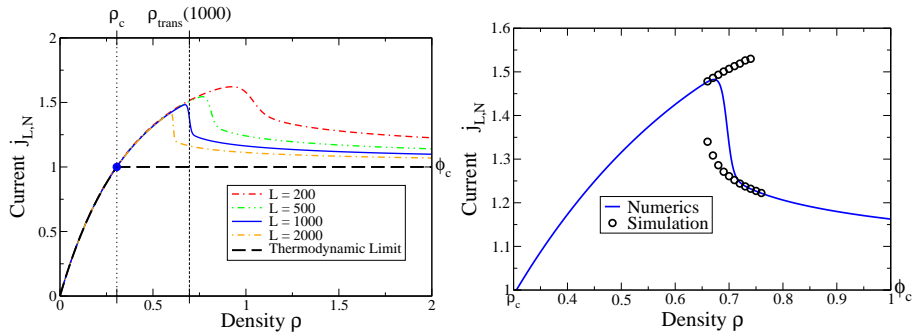

  \centering
  \includegraphics[width=0.48\textwidth]{numeric_curr_s05b4-Writeup}
  \includegraphics[width=0.49\textwidth]{MC_writeup}
  \caption{\label{fig:cur}  Finite size effects and current overshoot for $\gamma=0.5$ and $b=4$. Left: The canonical current for various system sizes as a function of the density $\rho=N/L$ are plotted. The dashed black line shows the thermodynamic current as a function of the system density $\rho$. 
  Right: The overshoot region for $L=1000$ showing the two distinct currents measured from Monte Carlo simulations for various densities near the maximum current.}
\end{figure}

\noindent In this section we present results obtained from exact numerics and Monte Carlo simulations in the canonical ensemble. 
We can calculate the canonical current given by \eq{CanCur} by making use of the following recursion relation for the canonical partition functions,
\begin{align}
  Z(L,N)=\sum_{k=0}^{N}w(k)Z(L-1,N-k).
\end{align}
Similarly we can calculate the canonical distribution of the maximum site occupation $M_L$.
To this end we define the cut-off canonical partition function which counts configurations for which $M_L(\feta)\leq m$,
\begin{align}
  Q(L,N,m)=\sum_{k=0}^{\min\{m, N\}}w(k)Q(L-1,N-k,m)\ .
\end{align}
This allows us to calculate
\begin{align}
  \label{eq:CanMax}
	\pi_{L,N}\left[M_L= m\right]=\frac{Q(L,N,m)-Q(L,N,m-1)}{Z(L,N)}\ ,
\end{align}
for all $m\in\N$ (the case $m=0$ is trivial). 
We often consider the equivalent formulation using background densities
\begin{align}
  \pi_{L,N}\left[M_L =  m\right] = \pi_{L,N}\left[ \frac{\SLbg}{L-1} = \rhobg \right] 
   \quad \textrm{ where } \quad \rhobg = \frac{N-m}{L-1}\ . \nonumber
\end{align}
This is more useful for illustrations and is more intuitive, since the background density characterises all but a single lattice site, while the formulation involving $M_L$ is more convenient for computations since it avoids issues with non-integer numbers. 
Therefore keeping both formulations in parallel is the best option for a concise presentation of our results.

\begin{figure}[t]
  \centering
  \includegraphics[width=0.48\textwidth]{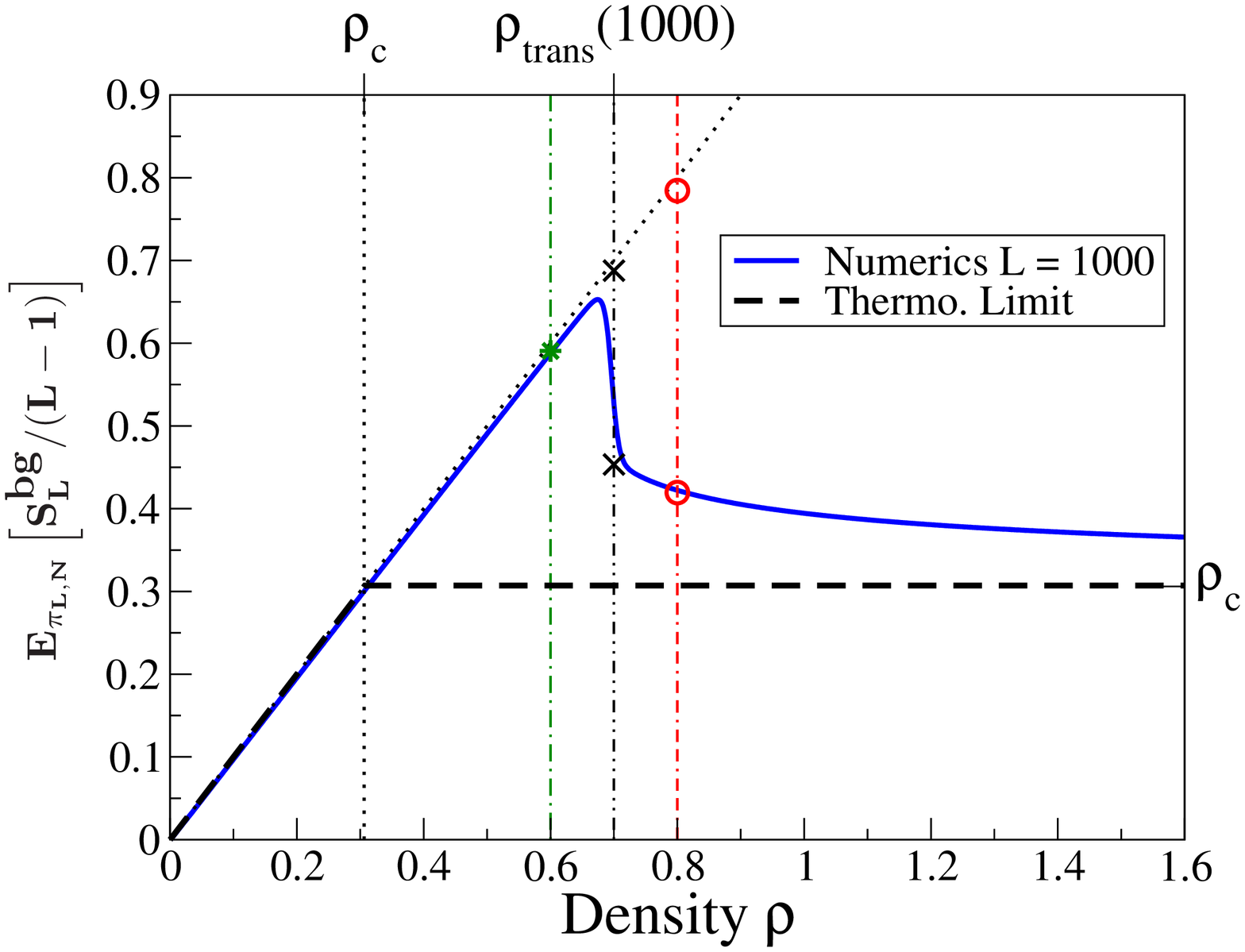}
  \includegraphics[width=0.49\textwidth]{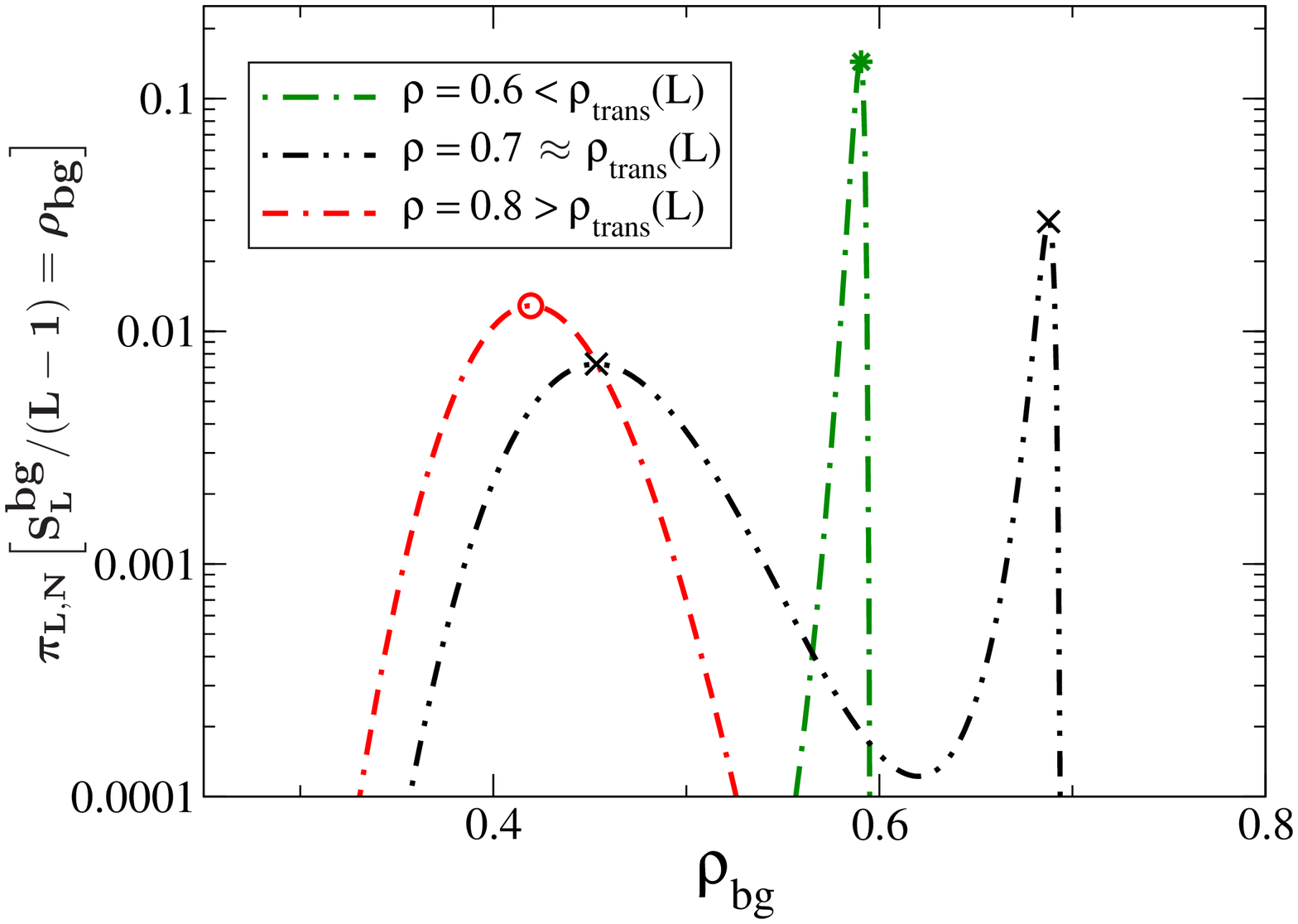}
  \caption{\label{fig:background} Finite size effects for the background density for $\gamma=0.5$, $b=4$ and $L=1000$. Left: The expected value of the background density. The dashed black line shows the thermodynamic limit result. Right: Distribution of the background density at three system densities shown by corresponding dash-dotted lines on the left, calculated exactly using \eq{CanMax}. The position of local maxima of the distributions are marked on both plots (\textcolor{green}{$\ast$} for $\rho=0.6$, $\times$ for $\rho=0.7$, \textcolor{red}{$\circ$} for $\rho=0.8$). The high background density maximum at $\rho=0.8$ occurs with extremely low probability and is off the scale.}
\end{figure}

On large finite systems we observe significant finite size effects above the critical density $\rho_c$, \fig{cur} shows the typical behaviour of the canonical current. 
Below $\rho_c$ the current is very close to the thermodynamic limit result even on relatively small systems ($L\sim100$) and the leading order finite size effects can be understood immediately from the proof of the thermodynamic limit result \cite{GrosskinskyCondensation}.
However the canonical current and background density significantly overshoot their critical values see \fig{cur} and \fig{background}.
An overshoot has been observed before for systems containing a single site defect \cite{AngelCondensation} for $\gamma = 1$.
For $\gamma<1$ the current increases monotonically with density $\rho = N/L$, in a way that appears to vary only very slightly with system size, up to some size dependent maximum current at $\rhotrans(L)$.
For $\rho < \rhotrans(L)$ the background density in the system is typically very close to $\rho$ (\fig{background}).
For a fixed system size $L$ we associate the region $\rho<\rhotrans(L)$ with a \textbf{putative fluid phase}.
Close to $\rho = \rhotrans(L)$ there is an abrupt decrease in the current and background density.
For $\rho > \rhotrans(L)$ the current and background density are both decreasing in $\rho$ and tend to their respective critical value as $\rho\to\infty$.
We associate the region $\rho > \rhotrans(L)$ with a \textbf{putative condensed phase} since the increasing density is entirely taken up by a large number of particles condensing on a single lattice site. 
The size of the effective fluid overshoot increases with increasing $b$ and with decreasing $\gamma$, and is already very pronounced  at $\gamma =0.5$, $b=4$.

\begin{figure}[t]
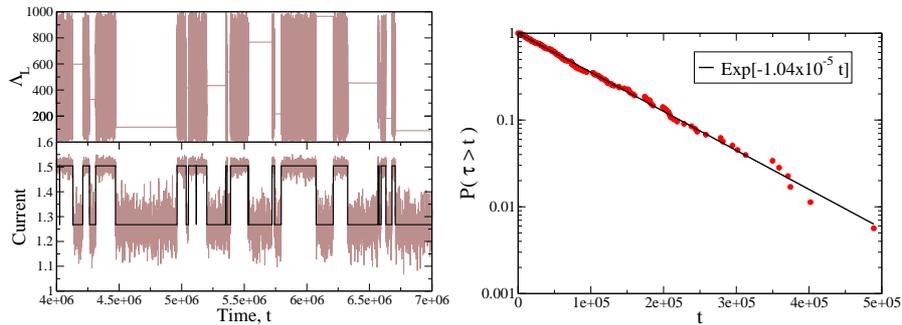

  \centering
  \includegraphics[width=0.48\textwidth]{switching_and_max-Writeup}
  \includegraphics[width=0.48\textwidth]{dist_fluid_s05b4-Writeup}
  \caption{\label{fig:switch} Switching dynamics for $\gamma = 0.5$, $b=4$ and $L=1000$ taken at $N/L=\rhotrans(1000)=0.695$. Left: (Bottom) Current against time from Monte Carlo simulations, calculated by taking the average number of jumps in the system over small time windows. The black line indicates the transition between the two putative phases. (Top) The location of the maximum does not change until the system is fluid.
Right: Cumulative tail of the distribution of waiting times in the putative fluid phase on a log linear scale. The solid line shows an exponential fit.}
\end{figure}

For $\rho=N/L$ close to $\rhotrans(L)$ we observe that the canonical distribution over background densities has two maxima of similar magnitude (\fig{background}) and Monte Carlo simulations show that the system switches between the two putative phases. 
\fig{switch} shows the typical behaviour of the current close to $\rhotrans(L)$ as a function of time, the putative phases can be clearly distinguished by the current. 
For $\rho\approx\rhotrans(L)$ in the condensed phase the location of the condensate does not change, while its position fluctuates heavily in the fluid phase, supporting the fact that the particles are distributed homogeneously.

The empirical distribution of waiting times in the two putative phases is very close to an exponential (\fig{switch} right). 
This suggests that the switching process is approximately Markovian over the range of parameters and jump distributions observed, and constitutes a genuine metastability phenomenon. 
The rate of the switching depends on the parameters $b$ and $\gamma$ as well as the jump distribution $p(x)$, which will be discussed in Section \ref{sec:dynamics} in more detail.

\section{Preliminary results}
\label{sec:heur}

\subsection{Heuristics and Current Matching}
As suggested by the form of the current overshoot in Fig 1, our approach is to approximate the putative fluid phase by extending the grand canonical current above $\phi_c$, which we achieve by means of a cut-off grand canonical measure. 
Since the total number of particles is fixed canonically to $N$ the system can not explore states where any single site contains more than $N$ particles.
We therefore expect the distribution of particles in the background under the canonical measure will always be closer (in any reasonable sense) to a grand canonical measure, with some suitably chosen cut-off and fugacitiy chosen to fix the correct density of particles, than it is to the unconditioned distribution.
For cut-off $m\in\N$ the \textbf{cut-off grand canonical measures} are defined by single site marginals with support on $\{0,1,\ldots,m\}$,
\begin{align}
\label{eq:CutMeas}
  \nu_{\phi,m}\left[\eta_x=n\right]&:=\nu_{\phi}\left[\eta_x= n|\eta_x\leq m \right] \nonumber\\
  &=\frac{1}{z_{m}(\phi)}w(n)\phi^{n} \quad \textrm{for} \quad n\leq m
\end{align}
where the normalisation is given by the finite sum
\begin{align}
  z_{m}(\phi)=\sum_{k=0}^{m}w(k)\phi^k.
\end{align}
These measures are well defined for all $\phi \in (0,\infty)$ and the current is given by
\begin{align}
  \label{eq:CutCur}
  \E{\nu_{\phi,m}}{g(\eta_x)} &= \phi\left( 1-\frac{w(m)\phi^{m}}{z_{m}(\phi)}\right).
\end{align}
Also the average density in the cut-off ensemble
\begin{align}
  \label{eq:CutDens}
  R_{m}(\phi):=\E{\nu_{\phi,m}}{\eta_x}=\frac{1}{z_{m}(\phi)}\sum_{k=0}^{m}k w(k) 
\end{align}
is a strictly increasing function from $[0,\infty)$ onto $[0,\infty)$ and so we denote its inverse
\begin{align}
  \label{eq:phicut}
  \Phi_{m}(\rhobg):=R_{m}^{-1}(\rhobg).
\end{align}
To a first approximation we estimate the current in the putative fluid phase using \eq{CutCur},
\begin{align}
  \label{eq:FluidCur}
  j_{L,N} \approx \Phi_N(\rho) 
\end{align}
under the assumption that $\left(1-w(N)\phi^{N}/z_{N}(\phi)\right)\approx 1$ for  values of $N$ under consideration. This approximation is shown in \fig{approxs} by the blue dashed line for $\gamma = 0.5 $, $b=4$ and a system size of $L=1000$.
The approximation is extremely close to the canonical current for $\rho<\rhotrans(L)$ and for $\rho\approx\rhotrans(L)$ is in good agreement with the empirically measured fluid currents from Monte Carlo simulations.
We demonstrate in the Section \ref{sec:rigour} that this estimate can be improved by choosing the cut-off more carefully.

\begin{figure}[t]
  \centering
  \includegraphics[width=0.7\textwidth]{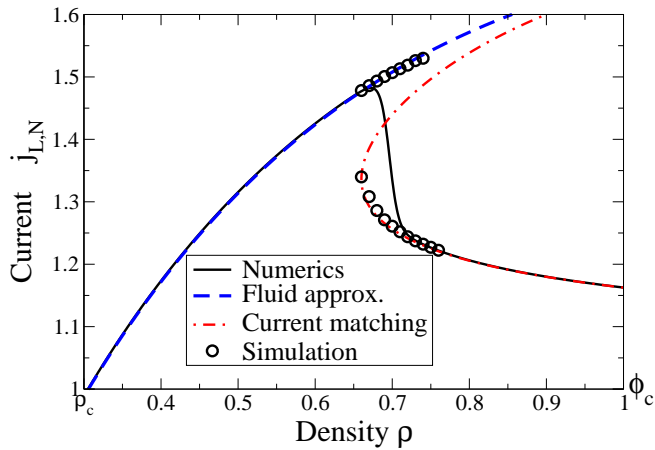}
  \caption{\label{fig:approxs} Current estimates above the critical point $\rho>\rho_c$ and $j_{L,N}>\phi_c$, for $L=1000$, $\gamma=0.5$ and $b=4$. The fluid approximation (\ref{eq:FluidCur}) and the current matching (\ref{eq:CurMatch}) agree well with canonical numerics and the metastable branches from Monte Carlo simulations (cf. Figs. \ref{fig:cur} and \ref{fig:switch}).}
\end{figure}

We approximate the current in the putative condensed phase by a current matching argument, a similar heuristic argument has been given before in \cite{KaupuzsZero-range}. 
The existence of a stable condensate implies that the average rate of particles exiting the condensate must be equal to the average rate of particles entering it. 
Conditioned on the occupation of the condensate $M_L =m$, the exit rate is simply $g(m)$ while the entry rate is well approximated by the stationary current in the background. 
This is assumed to be in a putative fluid phase described by $\nu_{\phi ,m}$, which leads to the current matching condition.
\begin{align}
  \label{eq:CurMatch}
  \Phi_m\left(\rhobg\right) = g(m) \quad \textrm{where}\quad \rhobg=\frac{N-m}{L-1}\ .
\end{align}
The lower branch of solutions to this equation define the condensed current approximation (\fig{approxs} red dashed line).
These are extremely close to the canonical current for $\rho>\rhotrans(L)$ and for $\rho\approx\rhotrans(L)$ are in good agreement with the empirically measured condensed currents from Monte Carlo simulations.

\subsection{Thermodynamic Limit  Rate Function}
In the thermodynamic limit we expect the existence of a rate function $I_{\rho}$ that describes the asymptotic probability of observing a background density $\rhobg\in (0,\rho)$. 
Precisely, we will consider the limit
\begin{align}
  \label{eq:mlimit}
  N,L,m\to\infty\quad\mbox{such that}\quad N/L\to\rho \ \textrm{and }\ \frac{m}{L}\to\rho - \rhobg\ , 
\end{align}
and expect that
\begin{align}
  \pi_{L,N}\left[ M_L = m \right]\sim e^{-L I_{\rho} (\rhobg)} \ .
\end{align}
To derive this we define the following function for finite systems
\begin{align}
  \label{eq:IDef}
  I_{L,N}(m)=-\frac{1}{L}\log \pi_{L,N}\left[ M_L = m \right]\ ,
\end{align}
which can be written as
\begin{align}
  \label{eq:IL}
  I_{L,N}(m) =&  -\frac{(L-1)}{L}\Big(\log z_{m}(\phi)- \frac{N-m}{L-1}\log \phi \Big)+ \nonumber\\
  & -\frac{1}{L}\log w(m)  +\frac{1}{L}\log Z(L,N)- \frac{1}{L}\log L+ \nonumber\\
  &  - \frac{1}{L}\log \nu_{\phi,m}^{L-1}\left[S_{L-1}{=}N-m \right]  {+} O(e^{-\frac{-b}{1-\gamma}m^{1-\gamma}})\ .
\end{align}
This is derived in Appendix \ref{app:deriv} and is valid for any $\phi\in(0,\infty)$. 
The first line in \eq{IL} resembles the thermodynamic entropy of the background (cf. \eq{thermoEnt}), the first term on the second line is the contribution due to the maximum occupied site.
The second term on the second line is the canonical normalisation and is independent of $m$. 
All the other terms will vanish in the limit (\ref{eq:mlimit}) and we get the following result.

\begin{theorem}
\label{thm:first}
  In the limit (\ref{eq:mlimit}) 
  $I_{L,N}(\rhobg)$ converges for all $\rhobg\in [0,\rho )$ as
  \begin{align}
    I_\rho(\rhobg):=\lim\limits_{\substack{L\to\infty}}I_{L,N}(m)=s(\rho) - s(\rhobg)\ ,
  \end{align}
  where $s$ is the thermodynamic entropy given by \eq{thermoEnt}.
  So for $\rho >\rho_c$ (supercritical case), $I_\rho(\rhobg)=0$ for each $\rhobg\geq \rho_c$ (see \fig{thermo}).
\end{theorem}

\begin{figure}[t]
  \centering
  \includegraphics[width=0.7\textwidth]{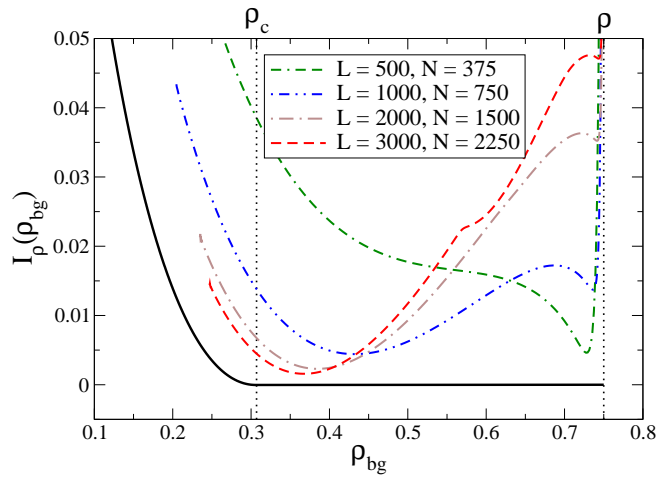}
  \caption{\label{fig:thermo} Thermodynamic limit rate function for a fixed system density $\rho=0.75$ with $\gamma=0.5$ and $b=4$. $I_{\rho}(\rhobg)$ is shown by a solid black line. The dashed lines show $I_{L,N}(m)$ against $\rhobg =(N-m)/(L-1)$ with $N/L=0.75$, for several values of $L$ calculated by exact numerics using \eq{CanMax}. Note convergence to $I_{\rho}(\rhobg)$ (solid black) is slow and initially non-monotonic above $\rho_c$.}
\end{figure}

\begin{proof}
  The proof follows directly from previous results on the equivalence of ensembles, for details see \cite{GrosskinskyCondensation}.
  It has been shown that
  \begin{align}
    \frac{1}{L}\log Z(L,N)\to s(\rho).
  \end{align}
  We use \eq{IL} and choose $\phi=\Phi(\rhobg)$ as defined by \eq{phi}.
  The cutoff $m\simeq N - \rhobg (L-1)$ diverges linearly in $L$ for all $\rhobg\in[0,\rho)$.
  It follows that for $\phi\leq\phi_c$ we have $z_{m}(\phi)\to z(\phi)$, and therefore with \eq{thermoEnt}
  \begin{align}
    \log z_{m}(\Phi(\rhobg))- \frac{N-m}{L-1}\log \Phi(\rhobg) \to 
    s(\rhobg)
  \end{align}
  It follows from the asymptotic behaviour of the single site weights \eq{AsymWk} that $\frac{1}{L}\log w(m)\to 0$, and the terms $\frac{1}{L}\log L$ and $O(e^{-\frac{-b}{1-\gamma}m^{1-\gamma}})$ also vanish in the limit \eq{mlimit}. 
  It remains to show that $\frac{1}{L}\log \nu_{\Phi(\rhobg),m}^{L-1}\left[S_{L-1}=N-m\right]\to 0$.
  In the 
  case $\rhobg\leq \rho_c$ we have $\Phi(\rhobg)\leq\phi_c$ and the first and second moments of $\nu_{\Phi(\rhobg),m}$ converge to those of $\nu_{\Phi(\rhobg)}$.
  The local limit theorem for triangular arrays (Thm. 1.2 in \cite{Daviselementary}) covers the sum, $S_L$, of independent random variables whose distribution depends on the number of terms $L$, which is the case here via the cut-off $m$.
  Convergence of the first two moments then implies
  \begin{align}
    \nu_{\Phi(\rhobg),m}^{L-1}\left[S_{L-1}=N-m\right]\sim \frac{1}{\sqrt{L}}\ ,
  \end{align}
  and the result follows immediately.
  The 
  case $\rhobg >\rho_c$ can be reduced to the previous case following \cite{GrosskinskyCondensation}, by arranging the excess mass $N-\rho_c (L-1)$ in the background among a finite number of sites so that each site contains at most $m$ particles. 
  This provides a sub-exponential lower bound for $\nu_{\phi_c,m}^{L-1}\left[S_{L-1}=N-m\right]$ which completes the proof.
\qed
\end{proof}
Previous results in the thermodynamic limit imply that for $N/L\to\rho>\rho_c$ the background density converges to $\rho_c$ (see \sect{thermo} \eq{Thermobg}).
Naturally, the thermodynamic limit result gives no indication of the sharp transition from fluid to condensed putative phases or the metastable switching observed on large finite systems.
For $\rho>\rho_c$ the system will appear to be condensed for sufficiently large $L$, since $\rho_{trans} (L)\searrow\rho_c$ as $L\to\infty$ (see \fig{cur}). Therefore the leading order finite size effects of Thm \ref{thm:first} do not describe the sharp transition or the metastability, as can be seen in \fig{thermo}. 
The apparent double well structure of $I_{L,N}$ is less pronounced for higher $L$, and the location of the global minimum shifts towards $\rho_c$. 
For $\rhobg =\rho$, i.e. $m/L\to 0$, the limit of $I_{L,N}$ actually depends on the precise scaling of m, which we do not discuss here (see [4] for more details).
To capture the sharp transition and the metastability we will replace the limit \eq{mlimit} by an appropriate scaling limit in section \ref{sec:rigour}.

\subsection{Heuristics on Metastability}
\label{subsec:heurMeta}

\begin{figure}[t]
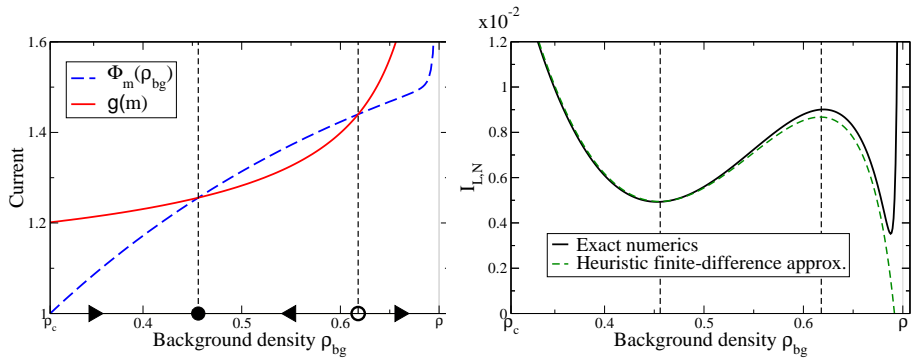

  \centering
  \includegraphics[width=0.48\textwidth]{currmatch_writeup}
  \includegraphics[width=0.49\textwidth]{canDist_writeup}
  \caption{\label{fig:curMatch}  Current matching on a finite system with $L=1000$ and $\rho=0.7$ close to $\rhotrans(1000)$ for $\gamma=0.5$, $b=4$. Left: The current in the background $\Phi_m(\rhobg)$ and current out of the maximum $g(m)$ plotted against $\rhobg=(N-m)/(L-1)$. Right: Exact numerics of $I_{L,N}$ using (\ref{eq:CanMax}) shown in solid black and approximation (\ref{eq:finiteDiff}) in dashed green. Vertical lines show the correspondence between current matching and critical points of $I_{L,N}$. The first local minimum corresponds to the lower stable branch of the current matching curve in \fig{approxs} and the local maximum corresponds to the upper unstable branch.  }
\end{figure}

\noindent We can understand the metastability on a heuristic level in terms of a simple current matching argument. 
$I_{L,N}(m)$ can be calculated efficiently for system sizes $L<4000$ using exact numerics (see \eq{CanMax}). 
In \fig{thermo} we see a clearly defined double well structure close to $\rhotrans(L)$ which accounts for the observed switching behaviour.
We may choose $\phi=\Phi_{m}(\rhobg)$ instead of $\phi=\Phi(\rhobg)$ in \eq{IL} so that the mean under the cut-off grand-canonical distribution is $\rhobg$.
The gradient of $I_{L,N}(m)$ is given by a one step finite difference.
If we assume for large $L$ that $\nu_{\Phi_{m}(\rhobg),m}^{L-1}\left[S_{L-1}/(L-1)\asymp \rhobg\right]\sim\frac{1}{\sqrt{2\pi L \sigma_c^2}}$ (which is valid in the scaling limit under certain conditions according to a local limit for triangular arrays (Thm. 1.2 in \cite{Daviselementary}), see Theorem \ref{thm:main} for details) and 
replace $m-1$ by $m$, then a straightforward calculation shows
\begin{align}
  \label{eq:finiteDiff}
  I_{L,N}\left(m-1\right) - I_{L,N}\left( m \right) \approx \frac{1}{L}\left(\log \Phi_m(\rhobg) - \log g(m)\right)\ ,
\end{align}
where $\rhobg=(N-m)/(L-1)$.
The first term on the right hand side follows in direct analogy with the thermodynamic limit for which $\partial_{\rho}s(\rho)=-\log\Phi(\rho)$ (cf.\ \cite{GrosskinskyEquivalence}), and the second term is 
the finite difference of the maximum site contribution using \eq{wn}.

This result holds rigorously in scaling limit and even for relatively small systems ($L\approx 1000$) $I_{L,N}$ is well approximated by \eq{finiteDiff} as is shown in \fig{curMatch}.
The approximation breaks down for background densities close to $\rho$ since the cut off $m$ becomes small and so the approximation $m-1\approx m$ is no longer valid. 
By the above argument we expect that solutions to the current matching equation (\ref{eq:CurMatch}) correspond to local maxima or minima of the rate function.
Whilst the current out of the maximum occupied site is greater than the current in the background we expect the background density to increase and vice versa.
Therefore the first point that the two currents cross is locally stable and so a local minimum of the rate function.
The next point they cross is a local maximum by the same argument, and there is another local minimum at a point close to $\rho$ associated with fluid configurations, which the current matching argument predicts to be at the boundary. 
This is shown in \fig{curMatch}.

\section{Rigorous Scaling Limit}
\label{sec:rigour}
In this section we explore the finite size effects by examining the leading order behaviour at the critical scale on which metastability persists.
As $L\to\infty$ this scale is  given by
\begin{align}
  \label{eq:critScale}
  N &= \rho_c L + \delta\rho L^{1-\alpha} + o(L^{1-\alpha}) \nonumber\\
  m &=(\delta\rho- \delta\rhobg) L^{1-\alpha} + o(L^{1-\alpha}) \quad \textrm{for}\quad \rhobg\in [0,\delta\rho)
\end{align}
for suitable $\alpha \in (0,1)$.
We may equivalently express it as
\begin{align*}
  N/L &= \rho_c + \delta\rho L^{-\alpha} + o(L^{-\alpha}) \\
  \frac{N-m}{L-1} &= \rho_c + \delta\rhobg L^{-\alpha} + o(L^{-\alpha})\ . \\
\end{align*}
The correct scaling exponent $\alpha$ to capture the canonical overshoot and metastability can be determined heuristically by our previous current matching argument \eq{CurMatch}.
If we assume that close to $\rho_c$ the fluid current \eq{FluidCur} can be approximated by the first term in the Taylor expansion around $\rho_c$, then $(\Phi_m (\rhobg)-1)\simeq \frac{\delta\rhobg}{\sigma_c^2}\, L^{-\alpha}$. 
The current matching equation (\ref{eq:CurMatch}) then implies
\begin{align*}
  \frac{\delta\rhobg}{\sigma_c^2}\, L^{-\alpha}\simeq g(m)-1\simeq b/m^\gamma =b(\delta\rho- \delta\rhobg)^{-\gamma} L^{-\gamma (1-\alpha )}\ ,
\end{align*}
which leads to $\alpha=\frac{\gamma}{1+\gamma}$.
In order to make this argument rigorous, also to find the transition point and describe the metastability, we study the canonical distribution of the background density (or equivalently the maximally occupied site) in the scaling limit. 

\subsection{The rate function}
We examine the asymptotic behaviour of $I_{L,N} (m)$ introduced in (\ref{eq:IDef}). 
We will see in Theorem \ref{thm:main} that there exists a unique $1>\beta>0$ such that $\lim\limits_{L\to\infty}L^{1-\beta}I_{L,N}(m)$ is finite and non-zero at the critical scale, we define
\begin{align}
  \label{eq:ScalingI}
  I^{(2)}_{\delta \rho}(\delta\rhobg):&=\lim\limits_{L\to\infty}L^{1-\beta}I_{L,{N}}(m)\ ,
\end{align}
where $N$ and $m$ are given by \eq{critScale}.
Since $I_\rho(\rhobg)=0$ for $\rhobg\geq\rho_c$ (see Sec.\ 4.2) this definition implies
\begin{align}
  \pi_{L,N}\left[ M_L =m\right]\sim e^{-L^\beta I^{(2)}_{\delta\rho} \left(\delta\rhobg\right)}
\end{align}
as $L\to\infty$, we recall that if $ M_L = m$ the background density is $\frac{N-m}{L-1}$ and $\delta\rho$ and $\delta\rhobg$ are given by \eq{critScale}. 
We begin by describing the scaling limit for the current in the fluid bulk and out of the maximally occupied site.
\begin{theorem}[Current scaling limits]
  \label{thm:curScale}
  For the scaling limit (\ref{eq:critScale}), with $\alpha=\frac{\gamma}{1+\gamma}$, the current out of the most occupied site and the average current in the fluid background are asymptotically given by
  \begin{align}
    g(m)=1+\frac{b}{(\delta\rho-\delta\rhobg)^\gamma}L^{-\alpha}+o(L^{-\alpha})\ ,\nonumber\\
    \Phi_{m}\left(\frac{N-m}{L-1}\right) =1 + \frac{1}{\sigma_c^2}\delta\rhobg L^{-\alpha}+o(L^{-\alpha}).\nonumber
  \end{align}
  provided $\frac{1}{\sigma_c^2}\delta\rhobg<\frac{b}{(1-\gamma)(\delta\rho-\delta\rhobg)^{\gamma}}$ .
\end{theorem}

\begin{proof}
  The result for the current out of the condensate $g(m)$ is immediate since $m = (\delta\rho-\delta\rhobg)L^{1-\alpha} + o(L^{1-\alpha})$, $g(m)= 1 + bm^{-\gamma}$ and $(1-\alpha)\gamma=\alpha$. 

  The proof of the fluid result follows from a Taylor expansion of the truncated density function $R_{m}(\phi)$ introduced in (\ref{eq:CutDens}), details on the expansion are contained in Lemmas \ref{Taylor} and \ref{convMoms} in the appendix.
  For convenience we write $\phi = e^{\mu}$, then
  \begin{align*}
    R_{m}(e^{\mu L^{-\alpha}})=\rho_{c}+\sigma_c^2\mu L^{-\alpha}+o(L^{-\alpha})
  \end{align*}
  for all $\mu \in [0,\frac{b}{(1-\gamma)(\delta\rho-\delta\rhobg)^{\gamma}})$.
  The result follows directly since
  \begin{align*}
    e^{\mu L^{-\alpha}} &= \Phi_{m}\left(R_{m}(e^{\mu L^{-\alpha}})\right)  \\
    \Rightarrow 1 + \mu L^{-\alpha} + o(L^{-\alpha}) &= \Phi_{m}\left( \rho_{c} + \sigma_c^{2}\mu L^{-\alpha}+o(L^{-\alpha}) \right)
  \end{align*}
  and choosing $\mu$ so that $\sigma_c^{2}\mu =\delta\rhobg$.
\qed
\end{proof}
Although the proofs in this section are restricted to $\frac{1}{\sigma_c^{2}}\delta\rhobg<\frac{b}{(1-\gamma)(\delta\rho-\delta\rhobg)}$ this includes the point of intersection of the two currents (cf.\ \fig{curMatch}) and we discuss a possible extension to the whole region $\delta\rhobg \in(0,\delta\rho)$ at the end of this Subsection. 
Theorem \ref{thm:curScale} states that in the scaling limit the average current in the fluid background and the jump rate out of the condensate converge to the critical current $\phi_c=1$ on the same scale as the density $N/L$ converges to $\rho_c$.
So for $N$ and $m$ at the critical scale, given by \eq{critScale}, $\left(1-\Phi_m(\rhobg)\right)L^{\alpha}$ and $\left(1-g(m)\right)L^{\alpha}$ converge to unique functions of the scaled variables $\delta\rho$ and $\delta\rhobg$.
We now show that in the scaling limit the rate function is exactly the integral of the difference in the two limiting currents (up to normalisation).

\begin{theorem}[Scaling limit rate function]
  \label{thm:main}
  In the scaling limit (\ref{eq:critScale}), with $\alpha=\frac{\gamma}{1+\gamma}$, the rescaled rate function (\ref{eq:ScalingI}) with $\beta=\frac{1-\gamma}{1+\gamma}$ converges as
  \begin{align*}
    \lefteqn{I^{(2)}_{\delta \rho}(\delta \rhobg) = \lim_{L\to\infty}L^{1-\beta}I_{L,N}(m)}\\
    & &= \frac{\delta\rhobg^2}{2\sigma_c^2} + \frac{b}{1-\gamma}(\delta\rho-\delta\rhobg)^{1-\gamma} -\inf_{r\in(0,\delta\rho)}\left\{\frac{r^2}{2\sigma_c^2}+ \frac{b}{1-\gamma}(\delta\rho-r)^{1-\gamma}\right\} 
  \end{align*}
  provided $\frac{1}{\sigma_c^2}\delta\rhobg<\frac{b}{(1-\gamma)(\delta\rho-\delta\rhobg)^{\gamma}}$.
\end{theorem}

\begin{proof}
  Firstly consider the unnormalised measure
  \begin{align}
    \label{eq:mu}
    \mu_{L,N} = \frac{Z(L,N)}{z^L(1)}\pi_{L,N}\ ,
  \end{align}
  for which it follows from definitions in Section \ref{subsec:statmeasure} and \eq{IDef}
  \begin{align}
    \frac{1}{L^\beta}\log\mu_{L,N}\left[ M_L = m\right] = -L^{1-\beta}I_{L,N}(m) +\frac{1}{L^{\beta}}\log Z(L,N) - L^{1-\beta}\log z(1)\ .
  \end{align}
  Re-writing $I_{L,N}$ using \eq{IL} we find
  \begin{align}
    \label{eq:scalingI}
    -\frac{1}{L^\beta} \log\mu_{L,N}\left[ M_L = m\right]  =&  -\frac{(L-1)}{L^{\beta}}\Big(\log z_{m}(\Phi_m(\rhobg^L))- \rhobg^L\log \Phi_m(\rhobg^L) \Big)+\nonumber \\
    & +L^{1-\beta}\log z(1)-\frac{1}{L^{\beta}}\log w(m) - \frac{1}{L^{\beta}}\log L+ \nonumber\\
    &  - \frac{1}{L^{\beta}}\log \nu_{\Phi_m(\rhobg^L),m}^{L-1}\left[S_{L-1}{=}N-m \right]  {-} o(1)\ ,
  \end{align} 
  where we used the shorthand $\rhobg^L =(N-m)/(L-1)$.
  By applying Lemma \ref{convMoms}, given in the appendix, the variance under the cut-off distribution $\nu_{\Phi_m(\rhobg^L),m}$ converges to the variance of the critical measure for $\frac{1}{\sigma_c^2}\delta\rhobg<\frac{b}{(1-\gamma)(\delta\rho-\delta\rhobg)^\gamma}$. 
  By the local limit theorem for triangular arrays (Theorem $1.2$ in \cite{Daviselementary} ) we have
  \begin{align}
    \label{eq:asymNu}
    \nu_{\Phi_{m}(\rhobg^L),m}^{L-1}\left[S_{L-1}=N-m\right] \sim \frac{1}{\sqrt{2\pi L \sigma_c^2}}
  \end{align}
  since $\rhobg^L\to\rho_c$ and $\E{\nu_{\Phi_{m}(\rhobg^L),m}}{\eta_x} \to \rho_c$ as $L\to\infty$ with the scaling (\ref{eq:critScale}).
  So the contribution of this term to (\ref{eq:scalingI}) vanishes.
  Also since $\beta>0$ we have $L^{-\beta}\log L \to 0$ as $L\to\infty$.

  The condensate contribution to \eq{scalingI} is given by 
  \begin{align}
    \label{eq:asymLogW}
    \frac{1}{L^{\beta}}\log w(m)&\simeq -\frac{1}{L^{\beta}}\frac{b}{1-\gamma}(\delta\rho-\delta\rhobg)^{1-\gamma}L^{(1-\alpha)(1-\gamma)}\nonumber\\
    &= -\frac{b}{1-\gamma}(\delta\rho-\delta\rhobg)^{1-\gamma}
  \end{align}
  where we have used the asymptotic behaviour of the single site weights (\eq{AsymWk}) and  $(1-\alpha)(1-\gamma)=\beta$.

  By applying Theorem \ref{thm:curScale} to $\Phi_m(\rhobg^L)$ and taking Taylor expansion of $z_m$ according to Lemma \ref{convMoms} (appendix),
   \begin{align}
    \label{eq:asymLogPhi}
    L^{1-\beta}\left( \rho_c + \delta\rhobg L^{-\alpha}\right)\log\Phi_{m}(\rho_c + \delta\rhobg L^{-\alpha}) &= \frac{\delta\rhobg}{\sigma_c^2}L^{\alpha}\rho_c  + \frac{\delta\rhobg^2}{\sigma_c^2} + o(1)
  \end{align}
  and
  \begin{align}
    \label{eq:asymLogz}
    \lefteqn{L^{1-\beta}\log z_{m}\left(\Phi_{m}(\rhobg^L)\right) =}\\ 
     &=L^{1-\beta}\log\left( z(1) + \frac{\delta\rhobg}{\sigma_c^2}L^{-\alpha}z'(1) + \frac{\delta\rhobg^2}{2\sigma_c^4}L^{-2\alpha}z''(1) +o(L^{-2\alpha})\right)\nonumber\\
     &=L^{1-\beta}\log z(1) + \frac{\delta\rhobg}{\sigma_c^2}L^{\alpha}\rho_c + \frac{\delta\rhobg^2}{2\sigma_c^2}+ o(1).
  \end{align}
  since $(1-\beta)=2\alpha$.
  Combining Equations (\ref{eq:asymNu}) to (\ref{eq:asymLogz}) in (\ref{eq:scalingI}) implies,
  \begin{align}
    \label{eq:asymMu}
    -\frac{1}{L^\beta}\log\mu_{L,N}\left[ M_L = m\right] = \frac{\delta\rhobg^2}{2\sigma_c^2}+ \frac{b}{1-\gamma}(\delta\rho-\delta\rhobg)^{1-\gamma}+o(1)
  \end{align}
  in the scaling limit (\ref{eq:critScale}) as $L\to\infty$, provided $\frac{1}{\sigma_c^2}\delta\rhobg<\frac{b}{(1-\gamma)(\delta\rho-\delta\rhobg)^{\gamma}}$.

  By definition
  \begin{align}
    Z(L,N) = z^L(1)\sum_{m=1}^{N}\mu_{L,N}\left[ M_L = m\right]. \nonumber
  \end{align}
  So we may bound $Z(L,N)$ as follows,
  \begin{align}
    z^L(1)\max_{m}\{\mu_{L,N}\left[ M_L = m\right]\} \leq Z(L,N) \leq N z^L(1)\max_{m}\{\mu_{L,N}\left[ M_L = m\right]\}. \nonumber
  \end{align}
  Applying \eq{asymMu} the asymptotic behaviour of $L^{-\beta}\log Z(L,N)$ in the scaling limit is given by
  \begin{align}
    \label{eq:Zln}
    L^{-\beta}\log Z(L,N) = L^{1-\beta}\log z(1) - \inf_{r\in (0,\delta\rho)}\left\{\frac{r^2}{2\sigma_c^2}+ \frac{b}{1-\gamma}(\delta\rho-r)^{1-\gamma}\right\}+o(1)\ ,
  \end{align}
  provided that the global minimum of $-\frac{1}{L^\beta}\log\mu_{L,N}\left[ M_L = m\right]$ is attained for $\frac{1}{\sigma_c^2}\delta\rhobg<\frac{b}{(1-\gamma)(\delta\rho-\delta\rhobg)^{\gamma}}$ in the scaling limit.
  This follows from results in \cite{ArmendarizZero}.
  Combining Equations (\ref{eq:asymMu}) and (\ref{eq:Zln}) the result follows from definitions \eq{IL} and \eq{mu}
  \begin{align*}
    L^{1-\beta}I_{L,N}(\rhobg) =& \frac{\delta\rhobg^2}{2\sigma_c^2}+ \frac{b}{1-\gamma}(\delta\rho-\delta\rhobg)^{1-\gamma}+\\
    &- \inf_{r\in (0,\delta\rho)}\left\{\frac{r^2}{2\sigma_c^2}+ \frac{b}{1-\gamma}(\delta\rho-r)^{1-\gamma}\right\} + o(1)\ .
  \end{align*}
\qed
\end{proof}

The result may break down for $\frac{1}{\sigma_c^2}\delta\rhobg>\frac{b}{(1-\gamma)(\delta\rho-\delta\rhobg)^{\gamma}}$ due to the large probability of a second sub-condensate forming on the same scale as the condensate.
The conditions under which this occurs could be found by conditioning on a maximum site occupation and using the same methods as for Theorem \ref{thm:main} to find the distribution over the second highest occupied site.
By using this method and iterating when necessary, the results here could be extended to the entire region $\delta\rhobg \in (0,\delta\rho)$. 
However our results already cover the relevant critical points of the rate function, the fluid minimum, condensed minimum and local maximum that constitutes a potential barrier.

\subsection{Current Matching and Overshoot}

Theorem \ref{thm:main} together with Theorem \ref{thm:curScale} imply that in the scaling limit given by (\ref{eq:critScale}) the derivative of the scaling rate function $I^{(2)}_{\delta\rho}$ is given by the difference in the current out of the condensate and the average current in the background. 
This is made precise in the following result, which is a rigorous version of the current matching argument of Section \ref{sec:heur}.
\begin{corollary}[Current matching]
  \label{cor:curMatch}
  \begin{align}
    \label{eq:curMatch}
    \partial_{\delta\rhobg}I^{(2)}_{\delta \rho}(\delta \rhobg)=\frac{\delta\rhobg}{\sigma_c^2}-\frac{b}{(\delta\rho-\delta\rhobg)^\gamma}
  \end{align}
  and  the two terms on the right-hand side are exactly the limiting current curves from Theorem \ref{thm:curScale} above the critical point.

\end{corollary}

This implies that
  \begin{align}
    \partial_{\delta\rhobg}I^{(2)}_{\delta \rho}(\delta \rhobg)\to -\infty \quad\mbox{as }\delta\rhobg\to\delta\rho\ ,
  \end{align}
so that $I^{(2)}_{\delta \rho}$ always exhibits a boundary minimum at $\delta\rhobg =\delta\rho$. 

Depending on the value of $\delta\rho$ the rescaled rate function has one of three qualitative forms characterised by the number of its extreme points, or equivalently the number of roots of (\ref{eq:curMatch}), and the position of the global minimum. 
With the threshold
\begin{align}
  c_0(\gamma,b)=\frac{1+\gamma}{\gamma}\big(\sigma_c^2 \gamma b\big)^{1/(1+\gamma )} \ ,
\end{align}
\eq{curMatch} has no real roots for $\delta\rho<c_0$, exactly one for $\delta\rho=c_0$, and two for $\delta\rho>c_0$. 
The latter correspond to a local minimum at $\delta\rhobg =r_0\in (0,\delta\rho)$ and a local maximum at $\delta\rhobg =r^*$, with $r_0 <r^*$ (cf.\ \fig{limit} on the right). 
As usual, minima of $I^{(2)}_{\delta \rho}$ correspond to metastable phases, and the depth of the local minimum at $r_0$ as compared to the one at the boundary determines which of the phases is stable (i.e. corresponds to the global minimum).
So $c_0$ marks the threshold above which (for $\delta\rho>c_0$) there exists a metastable condensed phase. 
Stability of the phases changes when both minima have the same depth, the density where this is the case is given by $\rho_{trans} (L)$ introduced in Section 3. 
In the scaling limit (\ref{eq:critScale}) this behaves as
\begin{align*}
  \rho_{trans} (L)=\rho_c +\delta\rho_{trans} L^{-\alpha} +o(L^{-\alpha} )\ ,
\end{align*}
and from Theorem \ref{thm:main} and Corollary \ref{cor:curMatch} we get the explicit expression
\begin{align}
  \delta\rhotrans = \sigma_c^{\frac{2}{1+\gamma}}(1+\gamma)(2\gamma)^{-\frac{\gamma}{1+\gamma}}\left( \frac{b}{1-\gamma}\right)^{\frac{1}{1+\gamma}}\ .
\end{align}
Note that both $c_0$ and $\delta\rho_{trans}$ are increasing with $\gamma$ and $b$. Their ratio simplifies to
\begin{align}
  \delta\rhotrans /c_0 =\frac12\Big(\frac{2}{1-\gamma}\Big)^{1/(1+\gamma)}\ ,
\end{align}
which also increases monotonically with $\gamma$.
In particular, $\delta\rho_{trans} >c_0$ for all $\gamma >0$ which implies existence of an extended metastability region.

We can now summarise the critical behaviour of the system at scale (\ref{eq:critScale}) in three cases, which are also illustrated in \fig{limit} (right).
\begin{case}[$\delta\rho<c_0$]
  There exists a unique minimum of the scaling rate function at $\delta\rhobg =\delta\rho$,
  \begin{align*}
    I^{(2)}_{\delta \rho}(\delta\rho) = 0\ ,
  \end{align*}
  corresponding to a stable fluid phase. 
  The most likely background density under the canonical measure is $\SLbg/(L-1)=\rho_c+\delta\rho L^{-\alpha}+o(L^{-\alpha})$.
\end{case}
\begin{case}[$c_0<\delta\rho<\delta\rhotrans$]
  There are two local minima of the scaling rate function
  \begin{align*}
    I^{(2)}_{\delta \rho}(r_0)>I^{(2)}_{\delta \rho}(\delta\rho) = 0\ ,
  \end{align*}
  corresponding to a stable fluid phase with $\delta\rhobg =\delta\rho$ and a metastable condensed phase with $\delta\rhobg =r_0$.
  The most likely background density under the canonical measure is still $\SLbg/(L-1)=\rho_c+\delta\rho L^{-\alpha}+o(L^{-\alpha})$.
\end{case}
\begin{case}[$\delta\rho > \delta\rhotrans$]
  There are two local minima of the scaling rate function
  \begin{align*}
    I^{(2)}_{\delta \rho}(\delta\rho) > I^{(2)}_{\delta \rho}(r_0)  = 0\ ,
  \end{align*}
  corresponding to a metastable fluid phase with $\delta\rhobg =\delta\rho$ and a stable condensed phase with $\delta\rhobg =r_0$.
  The most likely background density under the canonical measure is now $\SLbg/(L-1)=\rho_c+r_0 L^{-\alpha}+o(L^{-\alpha})$ and a finite fraction of the excess mass condenses on a single lattice site with $m=(\delta\rho-r_0)L^{1-\alpha}+o(L^{1-\alpha})$.
\end{case}

The value of $\delta\rhotrans$, and for $\delta\rho > \delta\rhotrans$ the position of the global minimum, agree with recent results on the distribution of the maximum at the critical scale \cite{ArmendarizZero}.
The expected value of the maximum site occupation (equivalently background density) under the canonical measure can be interpreted in terms of the stable solution of the current matching argument as described above.

\begin{figure}[t]
  \centering
  \includegraphics[width=0.48\textwidth]{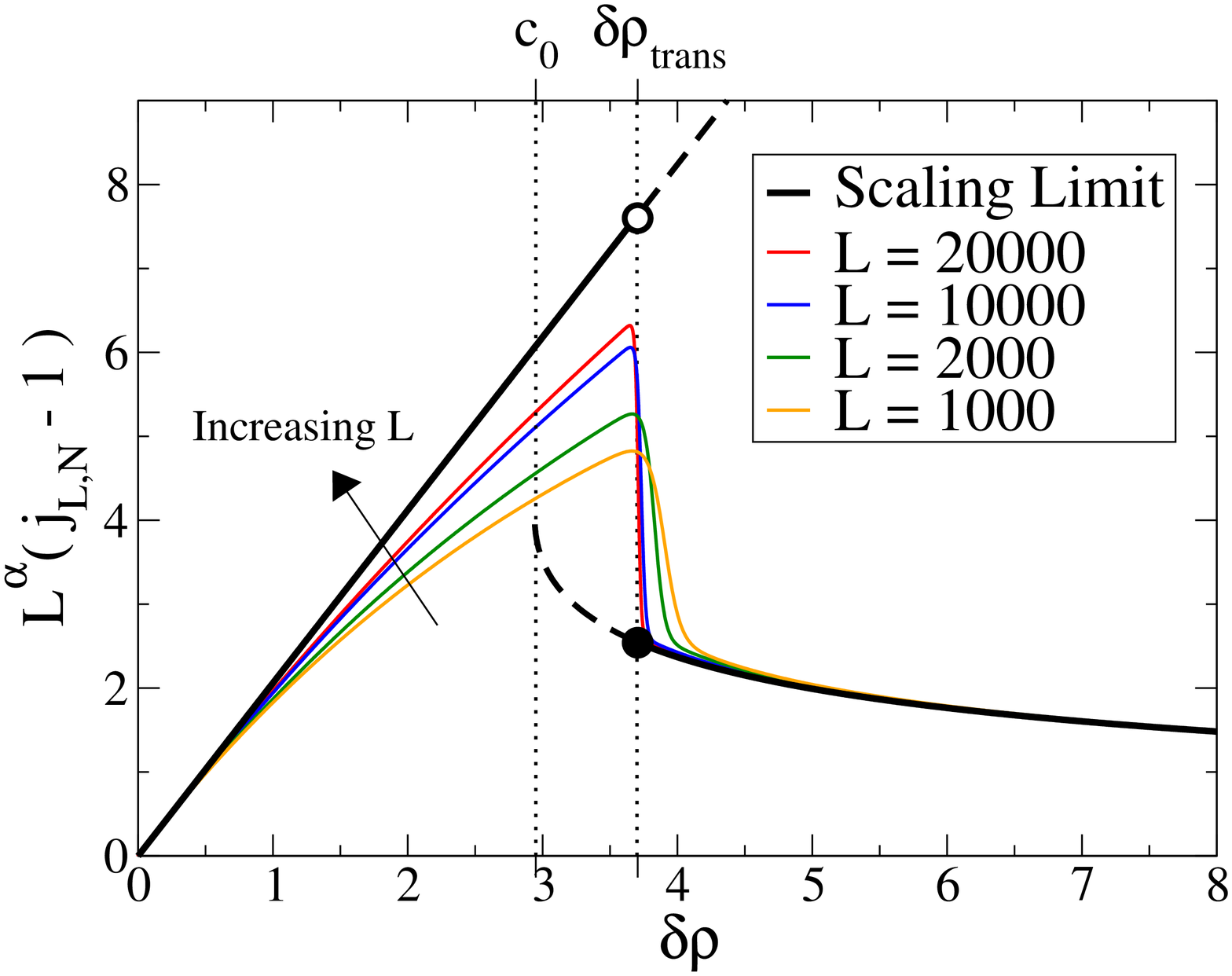}
  \includegraphics[width=0.48\textwidth]{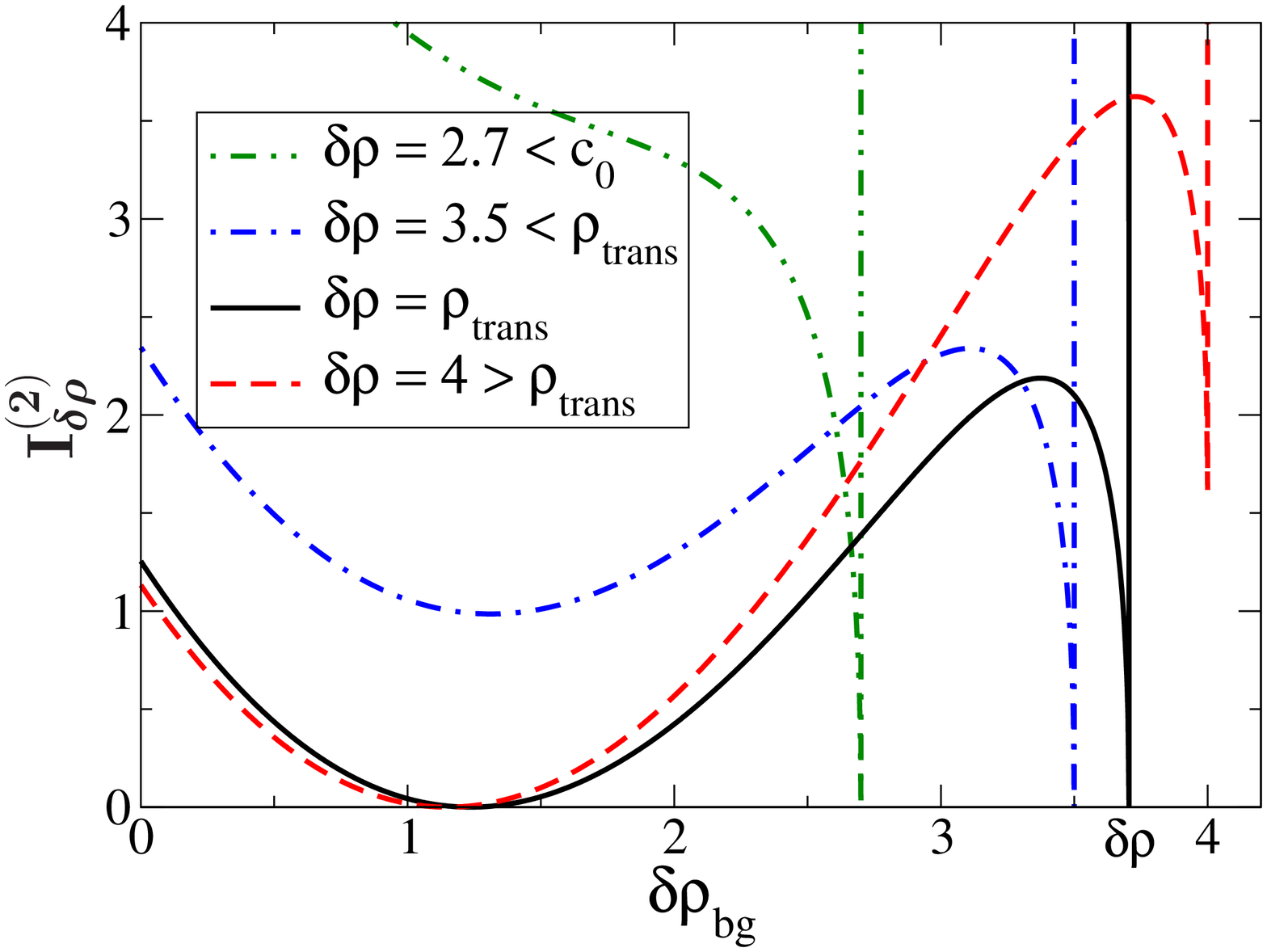}
  \caption{\label{fig:limit} Scaling limit for $\gamma=0.5$ and $b=4$. Left: The thick black line shows the scaling limit result including metastable branches (dashed). Thin coloured lines show the rescaled canonical currents from exact numerics, slowly approaching the scaling limit as $L$ increases.  The transition point is given by $\delta\rhotrans$. 
Right: $I^{(2)}_{\delta\rho}$ for various values of $\delta\rho$. Solid black corresponds to $\delta\rho=\delta\rhotrans$ where the depth of the two local minima are equal. The green dashed-dot-dot curve shows $I^{(2)}_{\delta\rho}$ for $\delta\rho<c_0$, blue dot-dashed curve corresponds to $c_0<\delta\rho<\delta\rhotrans$ and red dashed $\delta\rho>\delta\rhotrans$.
}
\end{figure}

\begin{corollary}[Canonical current overshoot]
  \label{cor:canCur}
  Under the conditions of Theorem \ref{thm:main}, for $\delta\rho<\delta\rhotrans$
  \begin{align}
    j_{L,N} - 1 = L^{-\alpha}\frac{\delta\rho}{\sigma_c^2}+o(1)
  \end{align}
  and for $\delta\rho>\delta\rhotrans$
  \begin{align}
    j_{L,N} - 1 = \frac{b L^{-\alpha}}{\left(\delta\rho-r_0\right)^{\gamma}} +o(1).
  \end{align}

\end{corollary}

\begin{proof}
  Follows directly by expressing $j_{L,N}$ in terms of the ratio of canonical partition functions \eq{CanCur} and applying \eq{Zln}.
\qed
\end{proof}

This result is illustrated in \fig{limit} (left). 
At the transition point $\delta\rho = \delta\rhotrans$ we must examine higher order terms to demonstrate that the system will asymptotically be in a condensed phase, see \cite{ArmendarizZero} and section \ref{sec:dynamics}.

\section{Estimating the Lifetime of Metastable Phases}
\label{sec:dynamics}

Having established the existence of metastable phases we can use our previous results to get an estimate of their lifetime.
Previous studies include results on the dynamics of fluctuations in a finite system \cite{GuptaFinite-size} and the dynamics of the condensate \cite{GodrecheDynamics,GodrecheDynamicsa,GrosskinskyCondensation}.
Our approach here follows mostly the one in \cite{GrosskinskyDiscontinuous} where metastability in a different zero-range process has been studied, and \cite{GodrecheDynamicsa} where a random walk argument was used to find the characteristic time of the motion of the condensate.

In this section we consider the dynamics in the scaling limit defined by \eq{critScale}. In Section 5.2 we have seen that for $\delta\rho >c_0$ the system exhibits two metastable phases,  we assume in the following that we are in this region. 
So for sufficiently large $L$ we know $I_{L,N}(m)$ exhibits a local maximum between the fluid and condensed minima. 
Let $m^*$ be the local maximum of $I_{L,N}$ which is well approximated by the closest integer to the largest root of \eq{CurMatch} (see Section \ref{sec:heur}).
In the scaling limit $\frac{N-m^*}{L-1}=\rho_c +r^* L^{-\alpha}$ (see Section \ref{sec:rigour}).
We define the metastable fluid (condensed) phase as all states for which the maximum site occupation is below (above) $m^{*}$, i.e. 
\begin{align}
  \label{eq:quasi}
  \pi^{\textrm{fluid}}_{L,N} &:= \pi_{L,N}[\cdot | M_L(\feta)\leq m^{*}] \nonumber\\
  \pi^{\textrm{cond}}_{L,N} &:= \pi_{L,N}[\cdot | M_L(\feta) > m^{*}]\ .
\end{align}
Since the process spends very little time close to the maximum of $I_{L,N}$, the precise choice of $m^*$ is not relevant for results in the scaling limit (cf. also \cite{BeltranMetastability} for a slightly different approach).
The lifetimes of the two metastable phases can be expressed in terms of the following hitting time,
\begin{align*}
  \tau_L := \inf\left\{t\geq 0 : M_L(\feta(t)) = m^*\right\}\ .
\end{align*}
We take
\begin{align}
  T_{\textrm{fluid}} := 2\E{\pi^{\textrm{fluid}}_{L,N}e^{\Lc t}}{\tau_L}\quad\mbox{and}\quad
  T_{\textrm{cond}} := 2\E{\pi^{\textrm{cond}}_{L,N}e^{\Lc t}}{\tau_L}\ ,
\end{align}
where the expectation is with respect to the dynamics \eq{infgen} with initial distribution given by $\pi^{\textrm{fluid}}_{N,L}$ or $\pi^{\textrm{cond}}_{N,L}$. 
The factor of 2 comes from the fact that once $m^*$ is reached the process can still return to the same metastable phase with probability $1/2$.

By ergodicity of the process $\big(\feta (t):t\geq 0\big)$ on $\Omega_{L,N}$ the ratio of lifetimes is directly related to the stationary distribution,
\begin{align}
  \label{eq:rationTauExact}
  \frac{T_{\textrm{fluid}}}{T_{\textrm{cond}}} = \frac{\pi_{L,N} [M_L(\feta)\leq m^{*}]}{\pi_{L,N} [M_L(\feta)> m^{*}]} = \frac{\sum_{k=0}^{m^*}e^{-L I_{L,N}(k)}}{\sum_{k=m^*+1}^{N}e^{-L I_{L,N}(k)}}\ .
\end{align}
Applying the scaling limit result of  Theorem \ref{thm:main}, we get for large $L$
\begin{align}
  \label{eq:rationTau}
  \frac{T_{\textrm{fluid}}}{T_{\textrm{cond}}} &\simeq \int_{r^*}^{\delta\rho}e^{-L^\beta I_{\delta\rho}^{(2)}(x)}\ud x \left( \int_{0}^{r^*}e^{-L^\beta I_{\delta\rho}^{(2)}(x)}\ud x \right)^{-1}\nonumber\\ 
  &\simeq\frac{\Gamma(1+\frac{1}{1-\gamma})\left(\frac{1-\gamma}{b}\right)^{\frac{1}{1-\gamma}} \sqrt{\partial^2_{x}I_{\delta\rho}^{(2)}(r_0)}}{{\sqrt{2\pi L}}}\, e^{-L^{\beta}(I_{\delta\rho}^{(2)}(\delta\rho)-I_{\delta\rho}^{(2)}(r_0))}\ ,
\end{align}
where the second line follows by a saddle point approximation of the denominator at the local minimum $r_0$ that corresponds to condensed configurations, and by expanding $I_{\delta\rho}^{(2)}(x)$ in the numerator at the boundary minimum (corresponding to fluid configurations), keeping the leading order singular term.
To leading order in the exponent this result implies that the ratio of the lifetimes is given by the relative depth of the minima, as expected. 
The lower order term $\sqrt{L}$ in the denominator implies that at the transition density $\delta\rhotrans$ in the scaling limit the condensed phase is stable and the fluid phase is metastable, in accordance with results in \cite{ArmendarizZero}.
However the scaling with $L$ is difficult to verify in Monte Carlo simulations, since system sizes have to be relatively small (at most of the order of $L=1000$) to get good statistics on switching times. In this regime higher order finite size effects still play a role and can affect the location of $\rhotrans (L)$ and therefore the relative depth of the minima of $I_{L,N}$.

\begin{table}[t]
 \centering
 \caption{\label{tab:1} Lifetime of metastable phases for $L=1000$ and $N=695$, with $\gamma=0.5$ and $b=4$. Using approximation (\ref{eq:lifetimes}) compared to measurements from Monte Carlo simulations on a fully connected lattice (MF), a one dimensional lattice with $p=1$ and with $p=3/4$.}
 \begin{tabular}[t]{l || c| c| c| c}
   \hline\noalign{\smallskip}
   & Approx. (\ref{eq:lifetimes}) & MF & 1D $p=1$ & 1D $p=3/4$ \\
   \noalign{\smallskip}\hline\hline\noalign{\smallskip}
   $T_{\textrm{fluid}}$ & $1.61\times10^{5}$ & $1.65\times10^{5}$& $1.21\times10^{5}$& $2.05\times10^{5}$ \\
   $T_{\textrm{cond}}$  & $1.51\times10^{5}$ & $1.52\times10^{5}$& $1.07\times10^{5}$& $1.94\times10^{5}$\\
 \end{tabular}
\end{table}

We estimate the two lifetimes by approximating the number of particles in the condensate by a continuous time random walk.
Since it is difficult to validate the scaling with $L$, as discussed above, we will demonstrate the validity of our approach by estimating the lifetimes directly for finite systems.
The actual lifetimes are not only related to the potential barrier given by the maximum of $I_{L,N}$, but also depend on the underlying dynamics.
The occupation of the maximum $\big( M_L (\feta (t)):t\geq 0\big)$ is a non-Markovian, stationary, ergodic process with state space $\Omega^* =\{ \lceil N/L \rceil,\lceil N/L \rceil + 1,\ldots,m^*,m^*+1,\ldots,N\}$ and stationary distribution $\pi^* (m) = e^{-L\, I_{L,N} (m)}$.
Since the process exhibits only single steps, it can be approximated by a continuous time random walk on $\Omega^*$ (the validity of this Markovian assumption is discussed later) where a particle leaves the maximum ($m\to m-1$) with rate $g(m)$. The rates corresponding to  $m\to m+1$ are then fixed by the stationary distribution $\pi^*$. 
In the scaling limit this implies that particles enter the condensate with rate given by the current in the background, see Corollary \ref{cor:curMatch}.
The lifetime of the metastable fluid phase is approximated by the mean first passage time of the random walk starting in state $\lceil N/L \rceil$ to reach $m^{*}$.
For the metastable condensed phase it is given by the mean first passage time starting at $N$, and both have a factor of 2 in front due to the possibility of reaching $m^*$ and not actually switching phase.
These can be calculated using standard techniques (see e.g. \cite{MurthyMean} for the discrete time analog),
\begin{align}
  \label{eq:lifetimes}
  T_{\textrm{fluid}} \approx &\ 2\sum_{i=\lceil N/L \rceil+1}^{m^{*}}\frac{\pi^*(i-1)}{g(i)\pi^*(i)}+ 2\sum_{i=\lceil N/L \rceil+2}^{m^{*}}\frac{1}{g(i)\pi^*(i)}\sum_{j=\lceil N/L \rceil}^{i-2}\pi^*(j) \nonumber\\
  T_{\textrm{cond}} \approx &\ 2\sum_{i=m^*+1}^{N}\frac{1}{g(i)} + 2\sum_{i=m^*+1}^{N}\frac{1}{g(i)\pi^*(i)}\sum_{j=i-1}^{N}\pi^*(j).
\end{align}
These two approximations can be calculated numerically for reasonably large systems (up to $L\approx 4000$) since we can calculate $\pi^* (m) = e^{-L\, I_{L,N} (m)}$ exactly using \eq{CanMax}.
Applying the scaling limit result and using saddle point approximations we recover an Arrhenius estimate with a prefactor for the lifetime of the two metastable phases,
\begin{align}
  T_{\textrm{fluid}} &\approx 2\int_{r^*}^{\delta\rho}e^{L^{\beta}I_{\delta\rho}^{(2)}(x)}\int_{x}^{\delta\rho}e^{-L^{\beta}I_{\delta\rho}^{(2)}(y)}\ud y \ud x \nonumber\\
  &\approx \Gamma\left(1+\frac{1}{1-\gamma}\right)\left( \frac{b}{1-\gamma}\right)^{\frac{-1}{1-\gamma}}\sqrt{\frac{2\pi L}{\partial^2_{r}I_{\delta\rho}^{(2)}(r^*)}}\, e^{L^{\beta}(I_{\delta\rho}^{(2)}(r^*)-I_{\delta\rho}^{(2)}(\delta\rho))}\nonumber\\
  T_{\textrm{cond}} &\approx 2\int_{0}^{r^*}e^{L^{\beta}I_{\delta\rho}^{(2)}(x)}\int_{0}^{x}e^{-L^{\beta}I_{\delta\rho}^{(2)}(y)}\ud y \ud x \nonumber\\ &\approx\frac{2\pi L}{\sqrt{\partial^2_{r}I_{\delta\rho}^{(2)}(r^*)\partial^2_{r}I_{\delta\rho}^{(2)}(r_0)}}\, e^{L^{\beta}(I_{\delta\rho}^{(2)}(r^*)-I_{\delta\rho}^{(2)}(r_0))}.
\end{align}
Here $r^*$ is the location of the local maximum of the rate function, $r_0$ is the local minimum corresponding to condensed configurations and $\delta\rho$ is the boundary minimum corresponding to fluid configurations.
To leading order the lifetime of the metastable fluid (condensed) phase is given by the height of the local maximum of the rate function (potential barrier) above the fluid (condensed) minimum.
This result agrees with the ratio of the lifetimes already discussed.
However as with the previous results the asymptotic form is difficult to validate due to higher order finite size effects, the exact expression \eq{lifetimes} is therefore more appropriate for relevant system sizes.

We expect the above random walk approximation to be accurate if particles exiting the maximum equilibrate in the bulk before returning to the condensate. 
This condition is best fulfilled for a fully connected lattice (mean-field geometry) and we expect it to be a reasonably good approximation for one dimensional totally asymmetric systems since particles have to pass through the whole fluid bulk before returning. 
For partial asymmetry and in higher dimensions a return without penetrating the fluid bulk is possible. 
To a first approximation the lifetimes have to be multiplied by the inverse probability of the event that particles escape into the bulk, since only such particles have a chance to equilibrate and contribute on the right scale. 
This pre-factor can be estimated using another random walk argument.
If a particle jumps to the right with probability $p$ and to the left with probability $q$ where $p+q=1$ and $p\neq 1/2$ then the probability of a particle reaching some macroscopic distance before returning to the condensate is asymptotically $|p-q|$.
We therefore expect the lifetime to increase by a factor of $1/|p-q|$ for partial asymmetry.
As a special case, for symmetric systems in one dimension this leads to an increase of lifetimes by a factor of $L$ by the same argument. 
For a detailed investigation and validation of this Markovian ansatz see \cite{GodrecheDynamicsa}. 
We find that this argument gives a good approximation for fully connected lattices, and totally asymmetric jumps in one dimension, see Table \ref{tab:1}. 
We observe that the one dimensional totally asymmetric case is a factor of approximately $1.4$ faster than the mean-field case, which is due to internal structure in the fluid background, and has also been observed and discussed in \cite{GodrecheDynamicsa}.
For partial asymmetry in one dimension with $p=3/4$ the process is slower than the totally asymmetric case by a factor of approximately $1.8$ supporting the arguments above which predict an increase by a factor of $2$.

The characteristic time for the motion of the condensate in the thermodynamic limit, $N/L\to\rho >\rho_c$, can be approximated by considering \eq{critScale} with a sequence of $\delta\rho$ increasing like $L^{\alpha}$.
The corresponding rescaled location of the condensed minimum $r_0\to 0$, since the condensed phase has limiting background density $\rho_c$. 
Our rigorous results do not technically hold in this limit, however our estimates of the lifetimes of the metastable phases agree with the characteristic times found by Godreche and Luck \cite{GodrecheDynamicsa,GodrecheDynamics}.
Our results demonstrate that close to the transition point in finite systems (i.e. $\delta\rho\approx \delta\rhotrans$) the condensate typically moves via the system entering the fluid phase (cf.\ \fig{switch}), and therefore its new position is expected to be chosen uniformly at random on $\Lambda_L$.
For $\delta\rho$ large, however, the lifetime of the fluid phase becomes small and the condensate can re-locate whilst the system remains in the condensed phase, via the mechanism described in \cite{GodrecheDynamicsa}.
A sub-condensate starts to grow in the fluid bulk, and the potential barrier the process has to cross for condensate motion is associated with the probability of the excess mass shared equally between two sites. This is the relevant mechanism in the thermodynamic limit.
For finite systems one can estimate which of these re-location processes dominates, by considering the distribution of the second largest site 
using the same techniques as in the previous section. 
It has recently been shown that the second mechanism via a sub-condensate can lead to a non-uniform relocation of the condensate \cite{BeltranMeta-stability}.


\section{Discussion}
\label{sec:discuss}
We have shown that a prototypical class of zero-range processes, that are known to undergo a condensation transition, exhibit large finite size effects including a metastable switching phenomenon for a large range of system parameters. 
We have characterized this behaviour using the background density as an order parameter, by rigorously deriving a large deviation rate function for its distribution in an appropriate scaling limit, which shows a double well structure. 
These results agree with recent work on the zero-range process at the critical scale \cite{ArmendarizZero}, which we extend by establishing metastable fluid and condensed phases. 
Our methods give rise to a simple interpretation of these results in terms of a stationary current balance between the fluid background and the condensate.
All results presented here, except those in Section \ref{sec:dynamics}, concern properties of the canonical stationary distribution and are therefore independent of the geometry or dimension of the lattice, so long as it permits homogeneous stationary distributions.

Our results allow us to estimate the lifetime of the two metastable phases using a heuristic random walk argument.
The estimates agree with previous studies of the dynamics of the condensate in the thermodynamic limit \cite{GodrecheDynamicsa}.
We show that in the critical scaling limit the lifetime of the two phases is growing with system size (as a stretched exponential) and derive the appropriate Arrhenius law including prefactor.
We have also demonstrated that on finite systems the re-location dynamics of the condensate is dominated by switching to the metastable fluid phase for a large parameter range, rather than the mechanism via a second sub-condensate assumed in previous results, which is relevant in the thermodynamic limit. 

\begin{figure}[t]
  \centering
  \includegraphics[width=0.48\textwidth]{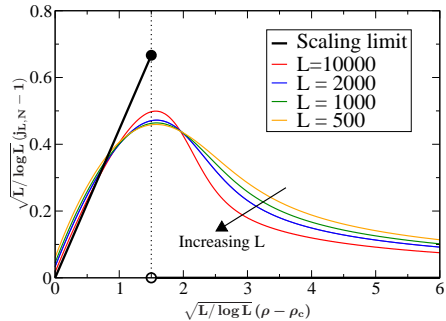}
  \caption{\label{fig:gammaOne}  The current overshoot in the scaling limit for $\gamma =1$ and $b=4$. The thick black line shows the scaling limit result and the rescaled canonical currents are shown for various system sizes by thin coloured lines, they slowly approaching the scaling limit as $L$ increases.  The transition point is given by $c_0(\gamma,b)$ which is given in \cite{ArmendarizZero}. }
\end{figure}

The zero-range process provides an interesting example where the thermodynamic limit fails to capture essential parts of the dynamics, and contradicts the usual expectation that finite systems should behave in a smoother fashion than the limiting prediction. 
Although the aim of this paper is not to discuss particular applications in detail, the finite size effects described here are particularly important for a proper understanding of various real condensation phenomena.
To name two relevant examples, the condensing zero-range process has recently been applied as a simplified traffic model as well as to describe clustering phenomena in granular media \cite{KaupuzsZero-range,ToeroekAnalytic,MeerCompartmentalized,LevineTraffic}.
Both of these applications exhibit metastability phenomena, and typical system sizes are of order $10^2 - 10^3$, the region where metastability effects in our analysis are most relevant. 
Since our results hold for a large generic class of zero-range models, they suggest that the apparent metastability in many applications is not necessarily the result of a particular choice of the jump rates, but rather a generic finite size phenomenon. 
While the results in the scaling limit hold for all zero-range processes with the same tail behaviour as the rates (\ref{eq:rates}), the strength of the effect for relevant system sizes will of course depend on the details of the particular application. 
Traffic modelling is an example where metastability is particularly pronounced \cite{WilsonMechanisms}, and it is an interesting question to investigate in detail the applicability of the zero-range process and our results in that area.

Using the methods discussed in this paper, a direct application of the findings in \cite{ArmendarizZero} allows us to extend our results on the current overshoot to zero-range processes with rates (\ref{eq:rates}) where $\gamma=1$ and $b>3$, see \fig{gammaOne}.
The critical scale in this case turns out to be of order $\sqrt{L/\log L}$, and the system also shows a discontinuous current overshoot in the scaling limit. However, as can be seen in \fig{gammaOne}, convergence is very slow and there are no metastable phases on the critical scale. 
Although metastability might occur on higher order scales, it will hardly be relevant in any finite size simulation or application. 

The scaling limit rate function in Theorem \ref{thm:main} could also be derived using a slightly modified version of the saddle point methods applied in \cite{EvansCondensation,EvansCanonical,MajumdarNature}. 
However, the direct approximation of the fluid and condensed metastable phases used in this paper seems more intuitive for our purpose, and can be made rigorous with less effort than a saddle point computation.

\section*{Acknowledgements}
The authors are grateful to Hugo Touchette, Gunter M.\ Sch{\"u}tz, Michalis Loulakis and Robin C.\ Ball for valuable discussions. This work was supported by the Engineering and Physical Sciences Research Council, UK, grant No. EP/E501311/1.

\appendix
\section*{Appendix}
\setcounter{section}{1}

\subsection{A formula for $I_{L,N}$}
\label{app:deriv}

In \eq{IDef} we introduce $I_{L,N}(m)=-\frac{1}{L}\log \pi_{L,N}\left[ M_L = m \right]$ for the canonical probability that the maximum site occupation is $m$. The following upper bound over counts configurations in which more than one site contains $m$ particles
\begin{align}
  \pi_{L,N}\left[ M_L= m \right] \leq \frac{L w(m) \sum_{\feta\in\hat{X}}\prod_{x=1}^{L-1}w(\eta_x)}{Z(L,N)}, 
\end{align}
where $\hat{X} = \{ \feta : \eta_1,\ldots,\eta_{L-1}\leq m, \sum_{x=1}^{L-1}\eta_x=N-m\}$.
The following lower bound does not count any configurations in which more than one site contains $m$ particles,
\begin{align}
  \pi_{L,N}\left[ M_L= m \right] \geq \frac{L w(m) \sum_{\feta\in\check{X}}\prod_{x=1}^{L-1}w(\eta_x)}{Z(L,N)}, 
\end{align}
where $\check{X} = \{ \feta : \eta_1,\ldots,\eta_{L-1}\leq m-1, \sum_{x=1}^{L-1}\eta_x=N-m\}$.
It follows that 
\begin{align}
  \pi_{L,N}\left[ M_L= m \right] = \frac{L w(m) \sum_{\feta\in\hat{X}}\prod_{x=1}^{L-1}w(\eta_x)}{Z(L,N)}\left(1-O(Le^{-\frac{-b}{1-\gamma}m^{1-\gamma}})\right)
\end{align}
as $L,m\to\infty$.
We may write the last term in the numerator in terms of the cut-off grand canonical measure and get 
\begin{align}
  &\pi_{L,N}\left[ M_L= m \right] = \nonumber\\ 
  &=\frac{L w(m) z_{m}^{L{-}1}(\phi)\phi^{{-}(N{-}m)}\nu_{\phi,m}^{L{-}1}\left[ S_{L{-}1}{=} N-m \right]}{Z(L,N)}\left(1{-}O(Le^{-\frac{{-}b}{1{-}\gamma}m^{1{-}\gamma}})\right)\ ,
\end{align}
which holds for all $\phi \in (0,\infty)$ (cf.\ Section 4.1). Finally, taking logarithm,
\begin{align}
  \label{eq:canapprox}
 -\frac{1}{L}\log \pi_{L,N}\left[ M_L= m \right] =&  -\frac{(L-1)}{L}\Big(\log z_{m}(\phi)- \frac{N-m}{L-1}\log \phi \Big)+ \nonumber\\
  & -\frac{1}{L}\log w(m)  +\frac{1}{L}\log Z(L,N)- \frac{1}{L}\log L+ \nonumber\\
  &  - \frac{1}{L}\log \nu_{\phi,m}^{L-1}\left[S_{L-1}{=}N-m\right]  {+} O(e^{-\frac{-b}{1-\gamma}m^{1-\gamma}}).
\end{align}

\subsection{Convergence of Moments under the Cut-off Ensemble}
\begin{lemma}
\label{Taylor}
  Consider two sequences $m_n\in\N$, $\mu_n\in(0,\infty)$,
  such that $m_n\to\infty$, $\mu_n\to 0$ and $m_n^\gamma\mu_n\to C$ as $n\to\infty$ with $C\in[0,\frac{b}{1-\gamma})$. Then for each $i\in\N_{0}$
  \begin{align}
    \limsup_{n\to\infty}\sum_{k\leq m_n}k^i w(k) e^{\mu_n k} < \infty. \nonumber
  \end{align}

\end{lemma}

\begin{proof}
  From the asymptotic behaviour of $w(k)$ \eq{AsymWk} we know there exists a $C_0\in (0,\infty)$ such that $w(k)\leq C_0 e^{\frac{-b}{1-\gamma}k^{1-\gamma}}$ for all $k$. So,
  \begin{align}
    \sum_{k\leq m_n}k^i w(k) e^{\mu_n k} \leq C_0 \sum_{k\leq m_n}k^i e^{\frac{-b}{1-\gamma}k^{1-\gamma}+\mu_n k}.
  \end{align}
  Fix $0 \leq \epsilon <(\frac{b}{1-\gamma} - C)$. Then there exists $\bar{n}\in\N$ such that $\mu_n\leq (C+\epsilon)m_n^{-\gamma}$ for $n\geq\bar{n}$.
 It follows that $\mu_n\leq (C+\epsilon)k^{-\gamma}$ for all $k\leq m_n$ and  $n\geq\bar{n}$.
 By applying this upper bound on $\mu_n$ we can bound the sum above uniformly in $n>\bar{n}$ as follows,
   \begin{align}
    \sum_{k\leq m_n}k^i w(k) e^{\mu_n k} &\leq C_0 \sum_{k\leq m_n}k^i e^{\frac{-b}{1-\gamma}k^{1-\gamma}+ (C+\epsilon)k^{1-\gamma}} \nonumber\\
    &\leq \int_{0}^{\infty}x^i e^{\frac{-b}{1-\gamma}x^{1-\gamma} + (C+\epsilon)x^{1-\gamma}} < \infty.
  \end{align}
\qed
\end{proof}

\begin{lemma}
  \label{convMoms}
  Under the assumptions of Lemma \ref{Taylor} we may bound the remainder in the Taylor expansion of each moment of the cut-off measure $\nu_{{\phi_n},{m_n}}$ introduced in (\ref{eq:CutMeas}). Writing $\phi_n=e^{\mu_n}$ we have
  \begin{align*}
    z_{m_n}(e^{\mu_n}) &= z(1) + z'(1) \mu_n  + o(\mu_n),\\
     R_{m_n}(e^{\mu_n}) &= \rho_c + \sigma_c^2\mu_n  + o(\mu_n),\\
     \textrm{and}\quad  \E{\nu_{{\mu_n},{m_n}}}{\eta_1^i}&= \E{\nu_{\mu_c}}{\eta_1^i} + \mu_n\partial_{\mu}\E{\nu_{\mu_c}}{\eta_1^i} +o(\mu_n).
\end{align*}
\end{lemma}

\begin{proof}
  For each $n\in\N$ the $i^{\textrm{th}}$ moment of the cut-off ensemble $\nu_{\phi_n,m_n}$ is given by a finite sum.
  We let,
  \begin{align}
    f_{m_n}(\mu) = \sum_{k\leq m_n}k^i w(k) e^{\mu k}.
  \end{align}
  then $f_{m_n}\in C^\infty(\R_{+},\R_{+})$.
  By Taylors theorem expanding around $\mu=0$,
  \begin{align*}
    f_{m_n}(\mu_n) &= f_{m_n}(0) + f_{m_n}'(0) \mu_n + h_n(\mu_n),
  \end{align*}
  where the remainder term may be expressed as
  \begin{align*}
    h_n(\mu_n) = \frac{f_{m_n}^{(2)}(s)}{2}\mu_n^2 \quad \textrm{for some} \quad s\in[0,\mu_n].
  \end{align*}
  $f_{m_n}^{(2)}$ is non-negative and increasing and so the remainder is bounded by
  \begin{align*}
    0\leq h_n(\mu_n) \leq \frac{f_{m_n}^{(2)}(\mu_n)}{2}\mu_n^2.
  \end{align*}
  By Lemma \ref{Taylor}, we know that there exists $U>0$ such that $f_{m_n}^{(2)}(e^{\mu_n})\leq U$ for all $n\in\N$.
  It follows that,
  \begin{align*}
     f_{m_n}(e^{\mu_n}) = f_{m_n}(1) + f_{m_n}'(1) \mu_n +o(\mu_n)
  \end{align*}
  The result then follows by considering $\left|f_{\infty}(1)-f_{m_n}(1)\right| \sim \sum_{k>m_n}k^i e^{\frac{-b}{1-\gamma}k^{1-\gamma}+\mu_n k}$ which is $o(\mu_n)$ since $\mu_n m_n^{\gamma}\to C$, so the tail of the sum converges as a stretched exponential, this allows us to replace the expected values under the cut-off with $\mu=0$ with the expected values under the true critical measure.
\qed
\end{proof}



\end{document}